\documentclass[journal,draftclsnofoot,onecolumn,12pt,twoside]{IEEEtranTCOM}

\normalsize

\ifCLASSINFOpdf

\else

\fi

\usepackage{amsmath}
\usepackage{tikz}
\usepackage{pgfplots}
\usepackage{pgf}
\usepackage{amssymb}
\usepackage{bm}
\usepackage{amsthm}
\usepackage{mathtools}
\usepackage{color, colortbl}
\newtheorem{theorem}{Theorem}
\newtheorem{definition}{Definition}
\newtheorem{lemma}{Lemma}
\usepackage[ruled,vlined]{algorithm2e}
\usepackage{cite}
\usepackage{multicol}
\usepackage{wrapfig}
\usepackage{floatrow}
\newfloatcommand{capbtabbox}{table}[][\FBwidth]

\usetikzlibrary{external}

\definecolor{myblue}{rgb}{0.00000,0.44700,0.74100}
\definecolor{myred}{rgb}{0.85000,0.32500,0.09800}
\definecolor{myred2}{rgb}{0.95000,0.32500,0.5}
\definecolor{mycolor3}{rgb}{0.92900,0.69400,0.12500}
\definecolor{mypurple}{rgb}{0.49400,0.18400,0.55600}
\definecolor{mygreen}{rgb}{0.46600,0.67400,0.18800}

\begin{document}

\title{Time-Bandwidth Product Perspective for Multi-Soliton Phase Modulation}

\author{Alexander Span, Vahid Aref, Henning B\"ulow, and Stephan ten Brink
\thanks{A. Span and S. ten Brink are with the Institute of Telecommunications, University of Stuttgart, Stuttgart, Germany}
\thanks{V. Aref and H. B\"ulow are with Nokia Bell Labs, Stuttgart, Germany.}
}

\markboth{}%
{}

\maketitle
\vspace{-20mm}
\begin{abstract}
Multi-soliton pulses are potential candidates for fiber optical transmission
where the information is modulated and recovered in the so-called nonlinear Fourier domain. 
While this is an elegant technique to account for the channel nonlinearity,
the obtained spectral efficiency, so far, is not competitive with classic Nyquist-based schemes. This is especially due to the observation that soliton pulses generally exhibit a large time-bandwidth product.
We consider the phase modulation of spectral amplitudes of higher order solitons, taking into account their varying spectral and temporal behavior when propagating along the fiber. For second and third order solitons, we numerically optimize the pulse shapes to minimize the time-bandwidth product.
We study the behavior of multi-soliton pulse duration and bandwidth and generally observe two corner cases where we approximate them analytically. We use these results to give an estimate on the minimal achievable time-bandwidth product per eigenvalue.
\end{abstract}


\IEEEpeerreviewmaketitle

\section{Introduction}

\IEEEPARstart{C}{oherent} optical technology has increased the transmission capacities of fiber channels to a point where the Kerr-nonlinearity becomes a limiting factor. Increasing transmit power leads to strong nonlinear distortions which need to be compensated. This is, however, quite complex and achieves only small gains, since the signal interacts nonlinearly with the noise along the channel. 

The Nonlinear Fourier Transform (NFT), that has recently regained attention \cite{yousefi2014nft}, shows a potential path to dealing with nonlinear fiber channels, as described by the Nonlinear Schroedinger Equation (NLSE). The NFT maps pulses to a nonlinear spectrum where the nonlinear crosstalk is basically absent, and the signal evolution along the fiber can be described by simple equations. This motivates to modulate data in this nonlinear spectral domain. The nonlinear spectrum consists of two parts: The \emph{continuous spectrum}, which is analogous to the conventional Fourier transform, and the \emph{discrete spectrum} which describes the solitonic component of a pulse. Multi-solitons of order $N$ ($N-$solitons), as considered in this paper, are special solutions to the NLSE with discrete nonlinear spectrum only, and can be described in the nonlinear spectrum by $N$ pairs of eigenvalue and a corresponding spectral amplitude.

Modulation of the continuous spectrum has been studied theoretically and experimentally in \cite{le2014nonlinear,le2016,Le2017b,prilepsky2014nonlinear,gemechu2017} and, so far, shows a higher spectral efficiency when compared to discrete spectrum modulation. Up to $2.3$ $\frac{\mathrm{bit}}{\mathrm{s \, Hz}}$ has been demonstrated in \cite{Le2018}.
Combined modulation of both the discrete and the continuous spectrum was demonstrated in \cite{Tavakkolnia2015,aref2016condis,Aref2018}. Discrete spectrum modulation has been intensively studied already two decades ago for on-off keying of first order solitons \cite{mollenauer2006solitons}. Multi-solitons have been proposed to increase the spectral efficiency \cite{yousefi2014nft}. Since then, eigenvalue and spectral amplitude modulation have been introduced and experimentally demonstrated. 

Characterization of the spectral efficiency for soliton transmission is, however, still an open problem. First, the noise statistics in the nonlinear spectrum are not fully understood. Second, there is generally no closed-form solution for multi-solitons and their pulse duration and bandwidth. Furthermore, the temporal and (linear) spectral properties may change due to modulation and during propagation. It is, therefore, not clear how to modulate the discrete spectrum in order to achieve a small time-bandwidth product. This is, however, important as the spectral efficiency will be proportional to the inverse of the time-bandwidth product. We have already investigated the behavior of pulse duration and bandwidth for multi-solitons of order $N=2,3$ with imaginary eigenvalues in \cite{Span2017}. We also have \emph{numerically optimized} their pulse shapes to minimize the time-bandwidth product in a spectral phase modulation scenario.
In this paper, we extend that work and give an \emph{analytical estimation} for general soliton orders and also consider eigenvalues with nonzero real part. We consider a scenario where the spectral amplitudes of $N$ predefined eigenvalues are independently phase modulated. We give a definition of the time-bandwidth product that takes into account the variations due to modulation and propagation along an ideal fiber. We give some general properties for multi-solitons that preserve the time-bandwidth product and reduce the number of optimization parameters. We numerically optimize the eigenvalue constellation and the absolute value of the corresponding spectral amplitudes to minimize this time-bandwidth product of the resulting multi-soliton pulses for $N \leq 3$. We distinguish the cases of purely imaginary eigenvalues or eigenvalues parallel to the real axis. Studying the behavior of multi-soliton pulse duration and bandwidth, we generally observe two limit cases, where either the pulse duration or the bandwidth becomes minimal, respectively. We derive analytical approximations for them in these limit cases. We use those results to estimate the minimal achievable time-bandwidth product for general soliton orders $N$. Numerical observations show that this may serve as a conjectured lower-bound. 

The paper is outlined as follows: In Sec.~\ref{sec:pre}, we review the basics of the NFT and some numerical algorithms. 
In Sec.~\ref{sec:multi_soliton_transmission}, recent works on soliton modulation are summarized. We then motivate the time-bandwidth product optimization for the soliton phase modulation scenario regarding spectral efficiency and give a reasonable definition for the time-bandwidth product. The soliton pulse optimization problem is statet in Sec.~\ref{sec:num_pulse_opt} and numerical optimization resulting in the ideal soliton pulses are presented. 
In Sec.~\ref{sec:TBP_higherN}, analytical approximations for the minimal achievable time-bandwidth product are derived. Conclusions are given in Sec.\ref{sec:conclusion}.

\section{Preliminaries - NFT and INFT}\label{sec:pre}

In this section, we briefly introduce the basics of NFT and its notation and show how multi-solitons can be constructed via the Inverse NFT (INFT) using the Darboux transform \cite{matveev1991darboux}.

The pulse propagation in single polarization along an ideal lossless and noiseless fiber is characterized by the normalized standard Nonlinear Schr{\"o}dinger Equation (NLSE)
\begin{equation}\label{NLSE}
\frac{\partial}{\partial z}q(t,z)+j\frac{\partial^2}{\partial t^2}q(t,z)+2j|q(t,z)|^2q(t,z)=0.
\end{equation}
The physical pulse $Q(\tau,\ell)$ at location $\ell$ along the fiber is then described by
\begin{equation}\label{eq:Q_physical}
Q\left(\tau,\ell\right)=\sqrt{P_0}\,\,\, q\left(\frac{\tau}{T_0},\ell \frac{\left|\beta_2\right|}{2T_0^2}\right) \text{ with } P_0\cdot T_0^2=\frac{\left|\beta_2\right|}{\gamma}, 
\end{equation}
where $\beta_2<0$ is the chromatic dispersion, $\gamma$ is the Kerr nonlinearity of the fiber, and $T_0$ determines the symbol rate. The closed-form solutions of the NLSE \eqref{NLSE} can be described in a nonlinear spectrum defined by the following so-called Zakharov-Shabat system~\cite{shabat1972exact}
\begin{equation}\label{eq:ZS}
\frac{\partial }{\partial t}\left(\begin{matrix}u_1(t;z)\\ u_2(t;z)\end{matrix}\right)=
	\left(\begin{matrix}
	0 & q\left(t,z\right) e^{2j\lambda t} \\-q^*\left(t,z\right) e^{-2j\lambda t} & 0
	\end{matrix}\right)
	\left(\begin{matrix}u_1(t;z)\\ u_2(t;z)\end{matrix}\right),
\qquad 
\lim_{t \to -\infty} \left(\begin{matrix}u_1(t;z)\\ u_2(t;z)\end{matrix}\right) \to \left(\begin{matrix}1 \\ 0 \end{matrix}\right)
\end{equation}
under the assumption that $q(t,z)\to 0$ decays sufficiently fast as $|t|\to \infty$ (faster than any polynomial).
The nonlinear Fourier coefficients (Jost coefficients) are defined as
\begin{align*}
a\left(\lambda;z\right) =\lim_{t\to \infty}u_1(t;z)
\qquad \qquad b\left(\lambda;z\right) =\lim_{t\to \infty}u_2(t;z).
\end{align*}
The set of simple roots of $a(\lambda;z)$ with positive 
imaginary part is denotes as $\Omega$. Theses roots are called \textit{eigenvalues} as they do not change in terms of $z$, i.e., $\lambda_k(z)=\lambda_k$. The nonlinear spectrum is usually described by the following two parts:

\begin{itemize}
\item[(i)] Continuous part: spectral amplitude ${Q_c(\lambda;z)=b(\lambda;z)/a(\lambda;z)}$ for real frequencies $\lambda\in\mathbb{R}$.
\item[(ii)] Discrete part: $\{\lambda_k,Q_d(\lambda_k;z)\}$ where 
$\lambda_k\in\Omega$, and $Q_d(\lambda_k;z)=b(\lambda_k;z)/\frac{\partial a(\lambda;z)}{\partial \lambda}|_{\lambda=\lambda_k}$.
\end{itemize}
Some algorithms to compute the nonlinear spectrum by numerically solving the Zakharov-Shabat system \eqref{eq:ZS} are reviewed in \cite{yousefi2014nft,wahls2015fast,SergeiK.Turitsyn2017}.

An $N-$soliton pulse is described by the discrete part only and the continuous part is equal to zero (for any $z$). 
The discrete part contains $N$ pairs of eigenvalue and the corresponding spectral amplitude, i.e., $\{\lambda_k,Q_d(\lambda_k;z)\},1\leq k\leq N $. We rewrite the eigenvalues as $\lambda_k=j\sigma_k+\omega_k$ with $\sigma_k\in \mathbb{R}^+,\omega_k \in \mathbb{R}$ and denote the spectral phases as $\varphi_k=\arg\left\{Q_d(\lambda_k)\right\}$. An important property of the nonlinear spectrum is its simple linear evolution \eqref{eq:Qd_evol} when the pulse propagates along the nonlinear fiber \cite{yousefi2014nft}. We define $Q_d(\lambda_k)=Q_d(\lambda_k;z=0)$. 
\begin{equation}\label{eq:Qd_evol}
Q_d(\lambda_k;z)=Q_d(\lambda_k)\exp(-4j\lambda_k^2z),
\end{equation}
\begin{algorithm}[tb]
\SetKwInOut{Input}{Input}
\SetKwInOut{Output}{Output}

\Input{Discrete Spectrum $\{\lambda_k,Q_d(\lambda_k)\}$; $k=1,\dots,N$}
\vspace{-3mm}
\Output{$N-$soliton waveform $q(t)$}
\vspace{-3mm}
\Begin{

\For{$k\leftarrow 1$ \KwTo $N$}{
\vspace*{-.7cm}
\begin{flalign}{\displaystyle
\rho_k^{(0)}(t)\longleftarrow\frac{Q_d(\lambda_k)}{\lambda_k-\lambda_k^*}\prod_{m=1,m\ne k}^N \frac{\lambda_k-\lambda_m}{\lambda_k-\lambda_m^*} e^{2j\lambda_k t}
}; &
\label{eq:Darboux_init}
\end{flalign}
\vspace*{-.9cm}
}

$q^{(0)} \longleftarrow 0$\;

\For{$j\leftarrow 0$ \KwTo $N-1$}{
\vspace*{-.7cm}
\begin{flalign}
q^{(j+1)}(t)\longleftarrow q^{(j)}(t)+2j(\lambda_{j+1}-\lambda_{j+1}^*)\frac{\rho_{j+1}^{*(j)}(t)}{1+|\rho_{j+1}^{(j)}(t)|^2}; &
\label{eq:sig_update}
\vspace*{-.9cm}
\end{flalign}

\For{$k\leftarrow j+2$ \KwTo $N$}{
\vspace*{-.7cm}
\begin{flalign}\rho_k^{(j+1)}(t) \longleftarrow \textstyle{
\frac{(\lambda_k -\lambda_{j+1})\rho_{k}^{(j)}(t)  +\frac{\lambda_{j+1}-\lambda_{j+1}^*}{1+|\rho_{j+1}^{(j)}(t)|^2} (\rho_{k}^{(j)}(t)-\rho_{j+1}^{(j)}(t))}
{\lambda_k -\lambda_{j+1}^*-\frac{\lambda_{j+1}-\lambda_{j+1}^*}{1+|\rho_{j+1}^{(j)}(t)|^2}\left(1 + \rho_{j+1}^{*(j)}(t)\rho_{k}^{(j)}(t) \right)}}; &
\label{eq:rho_update}
\end{flalign}
\vspace*{-.4cm}
}

}

}

\caption{INFT via Darboux transform variant in \cite{aref2016control}\label{alg:DT2}}

\end{algorithm}

This transformation is linear and depends only on its own eigenvalue $\lambda_k$. This property motivates the modulation of data over independently evolving spectral amplitudes.

The INFT is used to map the modulated nonlinear spectrum at the transmitter to the corresponding time domain pulse that is launched into the fiber. Some fast INFT algorithms can be found in \cite{wahls2015},\cite{wahls2016}. For the special case of the spectrum without the continuous part, the Darboux transform can generate the corresponding multi-soliton pulse~\cite{matveev1991darboux}. Alg.~\ref{alg:DT2} shows the pseudo-code of a variant of this transform\cite{aref2016control}.
It generates an $N-$soliton signal $q\left(t\right)$ by recursively
adding a pairs of $\{\lambda_k,Q_d(\lambda_k)\}$. 
The advantage of this algorithm is that it is exact with a low computational complexity and it can be used to express some basic properties of multi-soliton pulses.

\section{Multi-Soliton Transmission}\label{sec:multi_soliton_transmission}

\subsection{Motivation and Scenario}\label{sec:multi_soliton_motivation}

Soliton pulses are special solutions of the NLSE and are therefore matched to the nonlinear fiber channel. first order solitons have the special property of keeping their pulse shape during propagation along the (ideal) nonlinear fiber. Thus, they have a simple analytical description, and can easily be detected. 

The basic idea of eigenvalue modulation for $1-$solitons dates back to \cite{Hasegawa1993} where 1-out-of-$N$ imaginary eigenvalues were selected for transmission. A corresponding lower bound on the capacity per soliton with eigenvalue modulation was derived in \cite{Shevchenko2015} and a capacity achieving probabilistic eigenvalue shaping scheme was shown in \cite{Buchberger2018}. A noise model for the discrete spectrum was first introduced in \cite{zhang2015spectral}. 
By using higher order solitons and exploiting the spectral amplitudes, many follow-up studies aimed at designing modulation formats that allow to increase the number of bits being conveyed by a soliton.
Classical time domain modulation and detection by means of the corresponding eigenvalues was demonstrated in \cite{matsuda2014}. Other works directly modulated the discrete spectrum. On-off keying of up to ten predefined eigenvalues was demonstrated in 
\cite{dong2015nonlinear,aref2016onoff}. Multi-eigenvalue position encoding in \cite{hari2016multieigenvalue} achieved spectral efficiencies above $3$ $\frac{\mathrm{bit}}{\mathrm{s \, Hz}}$. However, only for short fiber lengths (compared to the dispersion length) or for dominating nonlinearity with small dispersion.
Exploiting the modulation of the spectral amplitudes has been demonstrated experimentally \cite{aref2015experimental,geisler2016} with up to $14$ bits per soliton \cite{buelow20167eigenvalues}. The noise statistics for the perturbation of the discrete spectrum have been investigated in \cite{Zhang2017,Buelow2018}. Correlations between different degrees of freedom in the nonlinear spectrum have been investigated also in \cite{Gui2017}, potentially allowing to further increase the amount of bits per soliton.

Under general conditions, all of these discrete spectrum-based transmission schemes imply the drawback of low spectral efficiency, owing to the, generally, large time-bandwidth product of soliton pulses. For $1-$solitons, their pulse shape is known to be "sech" in time and (linear) frequency domain. Although $N$-th order solitons have $N$ times more degrees of freedom for modulation and can, therefore, ideally carry $N$ times more information bits, they are not necessarily spectrally more efficient.
In most of the previous works, authors have either not focused on spectral efficiency or have investigated a "capacity per soliton", neglecting the time-bandwidth product of the soliton pulses.

Therefore our goal is to investigate the behavior of pulse duration and bandwidth for multi-soliton pulses in an ideal and noiseless transmission scenario. We consider the transmission of multi-solitons with $N$ predefined and fixed eigenvalues $\lambda_k$, $1\leq k \leq N$, along an ideal optical fiber (described by the NLSE), where each spectral amplitude $Q_d(\lambda_k)$ is independently phase modulated. The absolute values of the spectal amplitudes $|Q_d(\lambda_k)|$ are not modulated. Note that each spectral amplitude could carry an arbitrary QAM-constellation. Variation of $|Q_d(\lambda_k)|$, however, leads to pulse position modulation, may decrease spectrally efficiency. For this defined scenario, we optimize multi-soliton pulses to have a small time-bandwidth product.

The pulse shape of a multi-soliton not only depends on $\lambda_k$ and $|Q_d(\lambda_k)|$, but also on the phase combinations of $\varphi_k$. Modulating these phases, therefore, leads to different transmit pulses that can have different pulse durations $T$ and bandwidths $B$. Besides, the $Q_d(\lambda_k)$ are transformed during propagation according to \eqref{eq:Qd_evol}. Since all eigenvalues are necessarily distinct, the spectral amplitudes $Q_d(\lambda_k)$, and thus, pulse shape, pulse duration and bandwidth vary along the fiber. Such pulse changing effects do not occur for Nyquist-based signalling over linear channels.
Fig.~\ref{fig:soliton_propagation_example} illustrates two examples of multi-solitons and their propagation along the link. 
Fig.~\ref{fig:soliton_propagation_example} (a) shows a $3-$soliton with imaginary eigenvalues and Fig.~\ref{fig:soliton_propagation_example} (b) shows a $2-$ soliton with nonzero eigenvalue real part, each for two different initial phase combinations of the $\varphi_k$. The pulse shape is shown at the transmitter, once along the fiber and at the receiver. Fig.~\ref{fig:soliton_propagation_example} (c) and (d) show the evolution of the corresponding bandwidth $B$ and pulse duration $T$ during propagation.

\begin{figure}
\hspace{10mm}
\begin{minipage}{0.49\columnwidth}
\vspace{-3mm}
\begin{tikzpicture}[baseline=(current axis.south)]
\hspace{-20mm}
\begin{axis}[samples=50,ticklabel style={font=\footnotesize},width=0.4\textwidth,height=0.12\textheight,xmin=-12,xmax=12,ymin=0,ymax=2.5,y label style={yshift=-1em},label style={font=\small},ylabel={$T_\mathrm{max}B_\mathrm{max}$},xlabel={$z$},x label style={yshift=0.5em},legend style={font=\tiny},title style={font=\small},hide axis]     
 
    \addplot[forget plot,thick,myblue,mark=square,mark repeat=100,width=\linewidth]table {3Sol_example-1_Tx.dat};

\end{axis}
\hspace{20mm}
\begin{axis}[samples=50,ticklabel style={font=\footnotesize},width=0.4\textwidth,height=0.12\textheight,xmin=-12,xmax=12,ymin=0,ymax=2.5,y label style={yshift=-1em},label style={font=\small},ylabel={$T_\mathrm{max}B_\mathrm{max}$},xlabel={$z$},x label style={yshift=0.5em},legend style={font=\tiny},title style={font=\small},hide axis]     
 
    \addplot[forget plot,thick,myblue,mark=square,mark repeat=100,width=\linewidth]table {3Sol_example-1_mid.dat};                

\end{axis}
\hspace{20mm}
\begin{axis}[samples=50,ticklabel style={font=\footnotesize},width=0.4\textwidth,height=0.12\textheight,xmin=-12,xmax=12,ymin=0,ymax=2.5,y label style={yshift=-1em},label style={font=\small},ylabel={$T_\mathrm{max}B_\mathrm{max}$},xlabel={$z$},x label style={yshift=0.5em},legend style={font=\tiny},title style={font=\small},hide axis]     
 
    \addplot[forget plot,thick,myblue,mark=square,mark repeat=100,width=\linewidth]table {3Sol_example-1_Rx.dat};

\end{axis}

\end{tikzpicture}
\vspace{1mm}

\begin{tikzpicture}[baseline=(current axis.south)]
\hspace{-20mm}
\begin{axis}[samples=50,ticklabel style={font=\footnotesize},width=0.4\textwidth,height=0.12\textheight,xmin=-12,xmax=12,ymin=0,ymax=2.5,y label style={yshift=-1em},label style={font=\small},ylabel={$T_\mathrm{max}B_\mathrm{max}$},xlabel={$z$},x label style={yshift=0.5em},legend style={font=\tiny},title style={font=\small},hide axis]     
 
    \addplot[forget plot,thick,myred,mark=*,mark repeat=100,width=\linewidth]table {3Sol_example-2_Tx.dat};

\end{axis}
\hspace{20mm}
\begin{axis}[samples=50,ticklabel style={font=\footnotesize},width=0.4\textwidth,height=0.12\textheight,xmin=-12,xmax=12,ymin=0,ymax=2.5,y label style={yshift=-1em},label style={font=\small},ylabel={$T_\mathrm{max}B_\mathrm{max}$},xlabel={$z$},x label style={yshift=0.5em},legend style={font=\tiny},title style={font=\small},hide axis]     
 
    \addplot[forget plot,thick,myred,mark=*,mark repeat=100,width=\linewidth]table {3Sol_example-2_mid.dat};                  

\end{axis}
\hspace{20mm}
\begin{axis}[samples=50,ticklabel style={font=\footnotesize},width=0.4\textwidth,height=0.12\textheight,xmin=-12,xmax=12,ymin=0,ymax=2.5,y label style={yshift=-1em},label style={font=\small},ylabel={$T_\mathrm{max}B_\mathrm{max}$},xlabel={$z$},x label style={yshift=0.5em},legend style={font=\tiny},title style={font=\small},hide axis]     
           
    \addplot[forget plot,thick,myred,mark=*,mark repeat=100,width=\linewidth]table {3Sol_example-2_Rx.dat};

\end{axis}
\node at (2,0.7) {(a)};
\end{tikzpicture}
\vspace{1mm}

\begin{tikzpicture}[baseline=(current axis.south)]
\hspace{-20mm}
\begin{axis}[samples=50,ticklabel style={font=\footnotesize},width=0.4\textwidth,height=0.12\textheight,xmin=-12,xmax=12,ymin=0,ymax=2.5,y label style={yshift=-1em},label style={font=\small},ylabel={$T_\mathrm{max}B_\mathrm{max}$},xlabel={$z$},x label style={yshift=0.5em},legend style={font=\tiny},title style={font=\small},hide axis]

    \addplot[forget plot,thick,mygreen,mark=x,mark repeat=100,width=\linewidth]table {3Sol_example-3_Tx.dat};

\end{axis}
\hspace{20mm}
\begin{axis}[samples=50,ticklabel style={font=\footnotesize},width=0.4\textwidth,height=0.12\textheight,xmin=-12,xmax=12,ymin=0,ymax=2.5,y label style={yshift=-1em},label style={font=\small},ylabel={$T_\mathrm{max}B_\mathrm{max}$},xlabel={$z$},x label style={yshift=0.5em},legend style={font=\tiny},title style={font=\small},hide axis]     
 
    \addplot[forget plot,thick,mygreen,mark=x,mark repeat=100,width=\linewidth]table {3Sol_example-3_mid.dat};              

\end{axis}
\hspace{20mm}
\begin{axis}[samples=50,ticklabel style={font=\footnotesize},width=0.4\textwidth,height=0.12\textheight,xmin=-12,xmax=12,ymin=0,ymax=2.5,y label style={yshift=-1em},label style={font=\small},ylabel={$T_\mathrm{max}B_\mathrm{max}$},xlabel={$z$},x label style={yshift=0.5em},legend style={font=\tiny},title style={font=\small},hide axis]     
 
    \addplot[forget plot,thick,mygreen,mark=x,mark repeat=100,width=\linewidth]table {3Sol_example-3_Rx.dat};

\end{axis}

\end{tikzpicture}
\vspace{1mm}

\begin{tikzpicture}[baseline=(current axis.south)]
\hspace{-62mm}
\begin{axis}[samples=50,ticklabel style={font=\footnotesize},width=0.4\textwidth,height=0.12\textheight,xmin=-12,xmax=12,ymin=0,ymax=2.5,y label style={yshift=-1em},label style={font=\small},ylabel={$T_\mathrm{max}B_\mathrm{max}$},xlabel={$z$},x label style={yshift=0.5em},legend style={font=\tiny},title style={font=\small},hide axis]

    \addplot[forget plot,thick,mypurple,mark=diamond,mark repeat=100,width=\linewidth]table {3Sol_example-4_Tx.dat};

\end{axis}
\hspace{20mm}
\begin{axis}[samples=50,ticklabel style={font=\footnotesize},width=0.4\textwidth,height=0.12\textheight,xmin=-12,xmax=12,ymin=0,ymax=2.5,y label style={yshift=-1em},label style={font=\small},ylabel={$T_\mathrm{max}B_\mathrm{max}$},xlabel={$z$},x label style={yshift=0.5em},legend style={font=\tiny},title style={font=\small},hide axis]     
 
    \addplot[forget plot,thick,mypurple,mark=diamond,mark repeat=100,width=\linewidth]table {3Sol_example-4_mid.dat};                  

\end{axis}
\hspace{20mm}
\begin{axis}[samples=50,ticklabel style={font=\footnotesize},width=0.4\textwidth,height=0.12\textheight,xmin=-12,xmax=12,ymin=0,ymax=2.5,y label style={yshift=-1em},label style={font=\small},ylabel={$T_\mathrm{max}B_\mathrm{max}$},xlabel={$z$},x label style={yshift=0.5em},legend style={font=\tiny},title style={font=\small},hide axis]     
       
    \addplot[forget plot,thick,mypurple,mark=diamond,mark repeat=100,width=\linewidth]table {3Sol_example-4_Rx.dat};

\end{axis}
\node at (2,0.7) {(b)};
\draw[->] (-4.3,-0.5) to (2.1,-0.5);
\node at (-3.2,-0.75) {$z_0$};
\node at (-1.2,-0.75) {$z_1$};
\node at (0.9,-0.75) {$z_2$};
\end{tikzpicture}
\end{minipage}
\hspace{-33mm}
\begin{minipage}{0.49\columnwidth}
\begin{tikzpicture}[baseline=(current axis.south)]
\hspace{-1mm}
\begin{axis}[ticklabel style={font=\footnotesize},width=\textwidth,height=0.12\textheight,xmin=0,xmax=2.6,ymin=0.5,ymax=3,y label style={yshift=-0.5em},label style={font=\small},ylabel={Bandwidth $B$\\(a.u.)},ytick style={draw=none},yticklabels={,,},xlabel={Distance $z$},xtick={0,1.3,2.6},xticklabels={$z_0$,$z_1$,$z_2$},x label style={yshift=0.2em},y label style={align=center},legend style={font=\tiny},title style={font=\small},axis y line=left,axis x line=bottom]

    \addplot[forget plot,thick,myblue,mark=square,mark repeat=20,width=\linewidth]table {3Sol_example-1_Bw.dat}; 
    \addplot[forget plot,thick,myred,mark=*,mark repeat=20,mark phase=10,width=\linewidth]table {3Sol_example-2_Bw.dat};   
    \addplot[forget plot,thick,mygreen,mark=x,mark repeat=20,width=\linewidth]table {3Sol_example-3_Bw.dat}; 
    \addplot[forget plot,thick,mypurple,mark=diamond,mark repeat=20,mark phase=10,width=\linewidth]table {3Sol_example-4_Bw.dat};                 

\end{axis}
\node at (7,1) {(c)};
\end{tikzpicture}\vspace{-3mm}

\begin{tikzpicture}[baseline=(current axis.south)]
\hspace{-1mm}
\begin{axis}[ticklabel style={font=\footnotesize},width=\textwidth,height=0.12\textheight,xmin=0,xmax=2.6,ymin=8,ymax=16, label style={yshift=-0.5em},label style={font=\small},ylabel={Pulse Duration $T$\\(a.u.)},xlabel={Distance $z$},x label style={yshift=0.2em},ytick style={draw=none},yticklabels={,,},xtick={0,1.3,2.6},xticklabels={$z_0$,$z_1$,$z_2$},y label style={align=center},legend style={font=\tiny},title style={font=\small},axis y line=left,axis x line=bottom]

    \addplot[forget plot,thick,myblue,mark=square,mark repeat=20,width=\linewidth]table {3Sol_example-1_Tw.dat}; 
    \addplot[forget plot,thick,myred,mark=*,mark repeat=20,mark phase=10,width=\linewidth]table {3Sol_example-2_Tw.dat};   
    \addplot[forget plot,thick,mygreen,mark=x,mark repeat=20,width=\linewidth]table {3Sol_example-3_Tw.dat}; 
    \addplot[forget plot,thick,mypurple,mark=diamond,mark repeat=20,mark phase=10,width=\linewidth]table {3Sol_example-4_Tw.dat};                 

\end{axis}
\node at (7,0.4) {(d)};
\end{tikzpicture}
\end{minipage}

\caption{Multi-soliton phase modulation and transmission: (a) $3-$soliton with $\lambda_k = j \sigma_k$ and (b) $2-$soliton with $\lambda_k = j \sigma +\omega_k$ for two different phase combinations of $\varphi_k$, respectively. Pulses are shown at the transmitter, once along the fiber and at the receiver. 
The corresponding behavior of bandwidth $B$ and pulse duration $T$ along the fiber is shown in (c) and (d).}
\label{fig:soliton_propagation_example}
\end{figure}
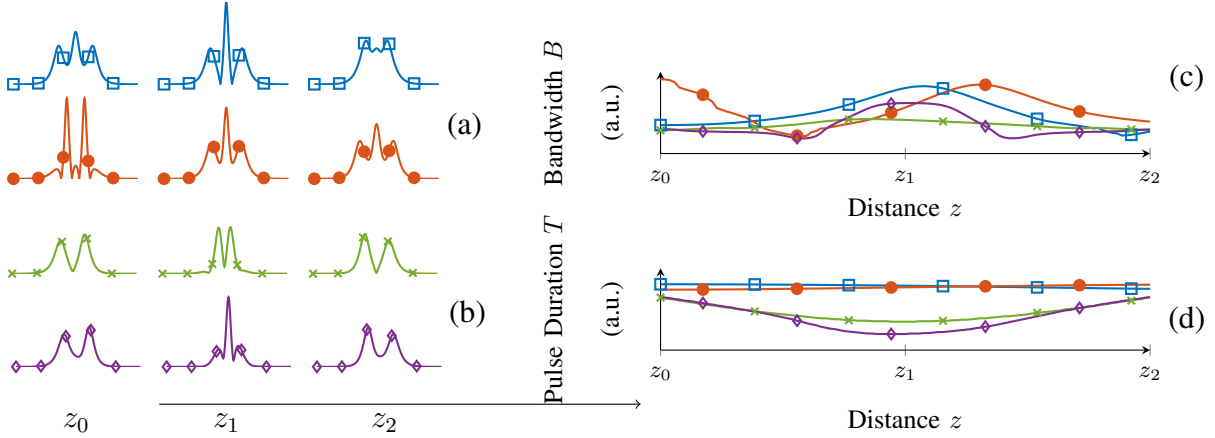

To characterize the spectral efficiency, we need to find a suitable definition for the time-bandwidth product of the soliton pulses. In an ideal tranmission case for a single channel scenario, this definition only needs to consider the bandwidth at the transmitter and the receiver $z \in \{0,L\}$ where the unknown variable $L$ is the link length. The bandwidth restriction is only given by the ability of transmitter and receiver to sample the signals. Assuming an ideal fiber transmission and amplification, the bandwidth (and its variation) along the link has no influence. For pulse duration, however, the evolution along the whole link $z \in [0,L]$ needs to be taken into account in order to avoid intersymbol interference in a train of soliton pulses. 

If the PSK constellation size for $Q_d(\lambda_k)$ is large enough, almost all phase combinations of $\varphi_k$ (and the different resulting bandwidths $B$ and pulse durations $T$) occur at the transmitter due to modulation. For simplification and to be independent of the phase modulation format, we thus consider the maximization of $T$ and $B$ over all existing combinations of $\varphi_k$. 
The absolute value of the spectral phases $|Q_d(\lambda_k)|$ is not modulated, however it may change as well during propagation according to \eqref{eq:Qd_evol}. Note that $|Q_d(\lambda_k)|$ remains fixed for purely imaginary eigenvalues. Based on these considerations, we define the time-bandwidth product as $\widehat{T} \cdot \widehat{B}$ with
\begin{align}\label{eq:TBP_def}
& \widehat{T}=\max_{\substack{\varphi_k, 1\leq k\leq N \\ |Q_d(\lambda_k,z)|, z \in [0,L]}} T
\qquad \qquad \qquad
\widehat{B}=\max_{\substack{\varphi_k, 1\leq k\leq N \\ |Q_d(\lambda_k,z)|, z \in \{0,L\}}} B
\end{align}
The pulse duration $T$ is maximized over all existing spectral phase combinations $\varphi_k$ and all spectral amplitudes $|Q_d(\lambda_k)|$ occuring \emph{at and between} transmitter and receiver. The bandwidth is maximized as well over all existing spectral phase combinations $\varphi_k$ but only over spectral amplitudes $|Q_d(\lambda_k)|$ occuring \emph{at} the positions of transmitter and receiver. Based on this definition, we aim to find multi-soliton pulses with minimum time-bandwidth product $\widehat{T}\cdot \widehat{B}$. Eigenvalues $\lambda_k$, the magnitude of the spectral amplitudes $\left|Q_d(\lambda_k)\right|$ at the transmitter, the transmission distance $L$ and the soliton order $N$ are left as parameters for this optimization. We observe that $\widehat{T} \cdot \widehat{B}$ is in general increasing with the soliton order $N$. However, a higher order soliton has also $N$ dimensions for encoding data. To have a fair comparison, we introduce the notion of "time-bandwidth product per eigenvalue" defined as
\begin{equation}\label{eq:TBP_per_eigenvalue}
\overline{T \cdot B}_N=\frac{\widehat{T} \cdot \widehat{B}}{N}.
\end{equation}

In the subsequent sections, we address the following questions:
\begin{itemize}
\item What is the minimum $\overline{T \cdot B}_N$ for a given $N$?

\item What are optimal choices for $\{\lambda_k\}_{k=1}^N$, $\{|Q_d(\lambda_k)|\}_{k=1}^N$ and how do optimal pulses look like?

\item Does the optimum "time-bandwidth product per eigenvalue" decrease with increasing $N$?
\end{itemize}

\subsection{Definition of Pulse Duration and Bandwidth}\label{sec:def}

Before we minimize \eqref{eq:TBP_per_eigenvalue}, the pulse duration $T$ and the bandwith $B$ need to be defined. One observes arbitrary multi-solitons having unbounded support and exponentially decaying tails in time and (linear) frequency domain. Their pulse shape depends on their nonlinear spectrum $\{\lambda_k,Q_d(\lambda_k;z)\}$ and changes by modulation and during propagation.
Despite pulse shape and peak power variation, the energy of the multi-soliton \emph{remains constant} and is equal to $\sum_{k=1}^N 4\sigma_k$.
This motivates the definition of pulse duration and bandwidth in terms of energy: 
\begin{definition}\label{def:TB_def_energy}

The pulse duration $T(\varepsilon)$ (and bandwidth
$B(\varepsilon)$, respectively) is defined as the smallest interval (frequency
band) containing $E_\mathrm{trunc} = (1 - \varepsilon)E_\mathrm{total}$ of the total soliton energy
\end{definition}

Note that truncation causes small perturbations of the eigenvalues. In practical transmission systems, the perturbations become even larger
due to intersymbol interference of adjacent truncated soliton pulses. Thus, there is a trade-off: $\varepsilon$ must be kept small such that the
truncation causes only small perturbations, but large enough
to have a small time-bandwidth product. In the following we use a practically reasonable $\varepsilon=10^{-4}$. Then, the time-bandwidth product of first order solitons becomes $\overline{T \cdot B}_1 \approx 9.9$. Another possible definition for $T$ and $B$ is: 

\begin{definition}\label{def:TB_def_threshold}

The pulse duration ${T(\alpha)=T_+ - T_-}$ and the bandwidth ${B(\alpha)=B_+ - B_-}$ are defined as the smallest interval such that 
\begin{multicols}{2}
\begin{itemize}
\item ${t\notin [T_-,T_+], \quad |q(t)|\leq \alpha \cdot \mathrm{max}\{|q(t)|\}}$ 
\item ${f\notin [B_-,B_+], \quad |Q(f)|\leq \alpha \cdot \mathrm{max}\{|Q(f)|\}}$
\end{itemize}
\end{multicols}
\end{definition}

$\alpha$ is chosen such that the resulting $T$ and $B$ are identical to the energy based definition for a first order soliton. Note that the peak $\mathrm{max}\{q(t)\}$ changes for different multi-soliton shapes.

\section{Multi-Soliton Pulse optimization}\label{sec:num_pulse_opt}

A first approach to minimize $\overline{T \cdot B}_N$ for a given $N$ is exhaustive search. For all combinations of $\{\lambda_k\}_{k=1}^N$, $\{|Q_d(\lambda_k)|\}_{k=1}^N$ and $L$, we evaluate $\widehat{T}\cdot \widehat{B}$
to finally select the optimum among them:
\begin{equation}\label{eq:TBP_minimization}
(\widehat{T}\cdot \widehat{B})^\star=\min_{\substack{\lambda_k,1\leq k\leq N \\ |Q_d(\lambda_k)|,1\leq k\leq N \\  L}} \widehat{T}\cdot \widehat{B}
\end{equation}
We find the optimum set of eigenvlaues $\{\lambda_k^\star\}$ and set of absolute value of spectral amplitudes at the transmitter $\{|Q_d^\star(\lambda_k)|\}$ such that the time-bandwidth product $\widehat{T}\cdot \widehat{B}$ defined by \eqref{eq:TBP_def} becomes minimal. For each given nonlinear spectrum, we construct the time domain signal $q(t)$ via Alg.~\ref{alg:DT2} and determine $T$ and $B$ numerically from $q(t)$ and its Fourier transform according to the definitions in Sec. \ref{sec:def}. We rewrite $Q_d(\lambda_k)$ as \eqref{eq:Qd_eta} where the amplitude scaling is $\eta_k \in \mathbb R^+$.
\begin{align}\label{eq:Qd_eta}
Q_d(\lambda_k) & =\eta_k \cdot |Q_\mathrm{d,init}(\lambda_k)| \cdot \exp(j\varphi_k)
\\
\mathrm{with} \qquad Q_\mathrm{d,init}(\lambda_k) & =\left(\lambda_k-\lambda_k^*\right) \prod_{m=1,m\ne k}^N \frac{\lambda_k-\lambda_m^*}{\lambda_k-\lambda_m}.\label{eq:Qd_0}
\end{align}
Thus, it is equivalent to optimize $\eta_k$ instead of $|Q_d(\lambda_k)|$. Note that the scaling factor $\eta_k$ is identical to the magnitude of the Fourier coefficient $|b(\lambda_k)|$. The benefit of this representation will become clear in the following section. We consider two special cases for the eigenvalues, either located on the imaginary axis with $\lambda_k=j\sigma_k$ or parallel to the real axis $\lambda_k=j\sigma+\omega_k$. W.l.o.g., we can assume ${\sigma_1 \geq \sigma_2 \geq \dots \geq \sigma_N}$. To indicate a "net"-gain when $N$ grows, we normalize $\overline{T \cdot B}_N$ by $\overline{T \cdot B}_1$, i.e. the time-bandwidth product of a first order soliton which can be derived analytically.

\subsection{Transformations Preserving the Time-Bandwidth Product}\label{sec:sol_prop}

To reduce the number of parameters to optimize, we show that the product $T\cdot B$ is preserved under the transformations following in Thm.~\ref{th:SolProperties}. The proofs can be found in Appx.~\ref{sec:sol_prop_proof}.

\begin{theorem}\label{th:SolProperties}
Assume that an $N$-soliton signal $q(t)$ corresponds to the nonlinear spectrum given by $\left\{\lambda_k=j\sigma_k+\omega_k\right\}_{k=1}^N$, $\left\{\eta_k\right\}_{k=1}^N$, $\left\{\varphi_k\right\}_{k=1}^N$. Then the following properties hold:

\begin{itemize}

\item[(i)]$\left\{\varphi_k-\varphi_0\right\}$ corresponds to $\exp(j\varphi_0) q(t)$.
\item[(ii)]$\left\{\exp(2\sigma_k t_0)\eta_k\right\}$, $\left\{\varphi_k-2\omega_k t_0\right\}$ corresponds to $q(t-t_0)$.
\item[(iii)] $\left\{\lambda_k / \sigma_0 \right\}$ corresponds to $\frac{1}{\sigma_0} q(t / \sigma_0)$ for $\sigma_0>0$.
\item[(iv)] $\left\{\omega_k-\omega_0\right\}$ leads to $\exp(2j\omega_0 t) q(t)$.
\item[(v)] $\left\{1/\eta_k\right\}$, $\left\{-\omega_k\right\}$ corresponds to $q(-t)$.
\item[(v*)] Setting $\left\{\eta_k=1\right\}$, $\left\{\omega_k=0\right\}$ leads to symmetric pulses ${q(t)=q(-t)}$.
\item[(vi)] $\{-\varphi_k\}$, $\{-\omega_k\}$ corresponds to $q^*(t)$
\end{itemize}

\end{theorem}

From these properties, simple restrictions on the optimization parameters that do not influence the time-bandwidth product of a multi-soliton \eqref{eq:TBP_minimization} can be assumed: (i) One can assume $\varphi_N=0$, (ii) one of the $\eta_k$ can be chosen arbitrarily, e.g. it suffices to assume $\eta_N = 1$, (iii) one of the $\sigma_k$ can be chosen arbitrarily, e.g. $\sigma_N=0.5$, (iv) one of the $\omega_k$ can be chosen arbitrarily, e.g. ${\omega_1=0}$ or ${\omega_1=-\omega_2}$, (v) it suffices to assume $\eta_2\in (0,1]$ or $\eta_2\in [1,\infty)$ when $\omega_k=0$ and (vi) the sign of one of the $\omega_k$ can be chosen arbitrarily, e.g. $\omega_1>0$.

\subsection{Eigenvalues on the Imaginary Axis}

First, we find optimum solitons with eigenvalues limited to the imaginary axis, $\lambda_k=j\sigma_k$. In this case, the transformation \eqref{eq:Qd_evol} simplifies to a phase rotation only; the $|Q_d(\lambda_k;z)|$ do not change during propagation. This means we can skip the maximization in terms of $|Q_d(\lambda_k;z)|$ along $z$ in \eqref{eq:TBP_def}.Thus the result will be independent of the link length $L$. Following properties i)-iii) and v) in Sec.~\ref{sec:sol_prop}, we choose $\varphi_N=0$, $\eta_N=1$, $\eta_{N-1}>1$ and $\sigma_N=0.5$. We define
\begin{equation}\label{eq:DeltaT_eta}
\Delta t_k= \ln \left( \eta_k \right) / 2\sigma_k,
\end{equation}
which, conceptually, is the temporal shift of the individual first order solitonic components by scaling $\eta_k$. We check all combinations of the free parameters $\sigma_k$ and $\Delta t_k$ to find the minimum $(\widehat{T} \cdot \widehat{B})^\star$ according to \eqref{eq:TBP_minimization}. Note that the solution is independent of the propagation distance $L$ in this case, since $|Q_d(\lambda_k;z)|$ does not change along the fiber. We uniformly quantize $\varphi_k$ to $m\cdot 2 \pi/128$, $m=1,\dots,128$. For each chosen combination of $\{\Delta t_k\}$ and $\{\sigma_k\}$ we calculate $\widehat{T}$ and $\widehat{B}$ in \eqref{eq:TBP_def} as the maximum of $T$ and $B$ for the resulting $128^{N-1}$ phase differences for $\{\varphi_k\}$. For a first coarse estimate of the minimal $\widehat{T} \cdot \widehat{B}$, the eigenvalues $\sigma_k$ are quantized in steps of $0.1$ in the range $[0.5,1.5]$, whereas we search $\Delta t_k$ in steps of $0.25$ in $[0,5]$ for $k=2$ and $[-5,5]$ for $k=3$. We increase the resolution around the observed minimum to $0.02$ for eigenvalues $\sigma_k$ and $0.05$ for $\Delta t_k$. Choosing the optimum values $\{\Delta t_k^\star\}$ and $\{\sigma_k^\star\}$ that achieve the minimal $(\widehat{T}\cdot \widehat{B})^\star$ results in the optimum soliton pulses that are shown in Fig. \ref{fig:optimum_pulse_2Sol_imag} for four different phase combinations of $\{\varphi_k\}$. Optimization is done for $N=2$, Fig.~\ref{fig:optimum_pulse_2Sol_imag}~(a) and $N=3$, Fig.~\ref{fig:optimum_pulse_2Sol_imag}~(b). The optimum values $\{\Delta t_k^\star\}$ and $\{\sigma_k^\star\}$ are given in Tab.~\ref{tab:imag_const_result} for the two definitions used for $T$ and $B$ (energy based and threshold based).
The optimum second and third order sets of pulses show an improvement in the time-bandwidth product per eigenvalue of $\overline{T \cdot B}_2 / \overline{T \cdot B}_1\approx 0.89$ and $\overline{T \cdot B}_3 / \overline{T \cdot B}_1\approx 0.84$ using the energy based definition for pulse duration and bandwidth described in Sec.~\ref{sec:def}. It turns out that the optimum pulses look similar to a train of first order solitons. The intuitive reason is that the nonlinear superposition of (overlapping) first order pulses causes large bandwidth expansions for some extreme choices of $\varphi_k$ due to the nonlinear interaction. Therefore, the first order components are slightly separated to reduce this bandwidth expansion with the cost of small increase in the pulse duration.

\begin{figure}
\begin{floatrow}
\resizebox{0.53\columnwidth}{!}{
\capbtabbox{%
\begin{tabular}{ c | c | c | c | c }    
    $N$ & definition & $\left\{\lambda_k^*\right\}_{k=1}^N$ & $\left\{\Delta t_k^*\right\}_{k=1}^N$ & $\frac{\overline{T \cdot B}_N}{\overline{T \cdot B}_1}$\\ \hline \hline
    \bm{$2$} & \bm{$\mathrm{energy}$} & \bm{$\{0.58; 0.5\}j$} & \bm{$\{2; 0\}$} & \bm{$0.89$} \\ 
    $2$ & $\mathrm{thr.}$ & $\{0.6; 0.5\}j$ & $\{1.4; 0\}$ & $0.86$ \\ 
    \bm{$3$} & \bm{$\mathrm{energy}$} & \bm{$\{0.7; 0.62; 0.5\}j$} & \bm{$\{-2.85; 1.05; 0\}$} & \bm{$0.84$} \\ 
    $3$ & $\mathrm{thr.}$ & $\{0.68; 0.62; 0.5\}j$ & $\{-2.75; 1.6; 0\}$ & $0.85$ \\
  \end{tabular}  
}{%
  \caption{Parameters for optimum pulses with imaginary eigenvalues.
  Bold parameters are used in Fig.~\ref{fig:optimum_pulse_2Sol_imag}}
   \label{tab:imag_const_result}
}
}\hspace{-3mm}
\resizebox{0.5\columnwidth}{!}{
\capbtabbox{%
  \begin{tabular}{ c | c | c | c | c }
    
    $N$ & definition & $\omega_k^\star$ & $\left\{\Delta t_k^*\right\}$ & $\frac{\overline{T \cdot B}_N}{\overline{T \cdot B}_1}$ \\ \hline \hline
    2 & energy & $\{0.075; -0.075\}$ & $\{-0.9; 0.9\}$ & 0.74 \\ 
    2 & thr. & $\{0.22; -0.22\}$ & $\{-0.9; 0.9\}$ & 0.75 \\       
    \bm{$3$} & \bm{$\mathrm{energy}$} & \bm{$\{0.55; -0.55; 0\}$} & \bm{$\{-2.2; 2.2; 0\}$} & \bm{$0.71$} \\  
    3 & thr. & $\{0.14; -0.14; 0\}$ & $\{-2.13; 2.13; 0\}$ &  0.71 \\ 		
  \end{tabular}
}{%
  \caption{Parameters for optimum pulses with eigenvalues parallel to the real axis $\lambda_k=0.5j+\omega_k$. Bold parameters are used in Fig.~\ref{fig:optimum_pulse_3Sol_real}}
  \label{tab:real_const_result}  
}
}
\end{floatrow}

\end{figure}

\subsection{Eigenvalue Real Part Constellation}\label{sec:real_const}
Next, we consider the pulse optimization for solitons with eigenvalues parallel to the real axis $\lambda_k=j\sigma+\omega_k$. We check all combinations of the free parameters $\omega_k$, $\Delta t_k$ to find the minimum $(\widehat{T} \cdot \widehat{B})^\star$ according to \eqref{eq:TBP_minimization}. Following property i) and iii), we choose ${\varphi_1=0}$ and $\sigma=0.5$. Following property iv), one can set $\omega_1=-\omega_2$ for the optimization of the time-bandwidth product $\hat{T}\cdot \hat{B}$. Considering property ii) for the given specific case $\sigma_k=\sigma$ one can further choose $\eta_2=1/\eta_1$ which is identical to $\Delta t_2=-\Delta t_1$ (see \eqref{eq:DeltaT_eta}). Note that property iv) and ii) allow also other choices; all $\omega_k$ could be shifted by some $\omega_0$ or all $\eta_k$ could be scaled by some $\exp(2\sigma t_0)$ without affecting $T \cdot B$  \emph{at a specific} $z$. For the real part eigenvalue constellation, however, $|Q_d(\lambda_k)|$, and thus $\Delta t_k$ are changing during propagation (see \eqref{eq:Qd_evol}). The $\omega_k \neq 0$ lead to a frequency shift of the first order solitonic component associated to this eigenvalue. This leads to solitonic components with different propagation velocities proportional to $\omega_k$. Consequently, multi-solitons with distinct $\omega_k$ will eventually be separated into their first order components during propagation as $z\to \infty$. This is associated to an increased required time frame for each pulse. According to property vi), we can set the signs of $\omega_1$ and $\omega_2$ depending on the signs of $\Delta t_1$ and $\Delta t_2$ such that the solitonic components will move slowly towards each other, collide at the soliton pulse center and then separate again during propagation \cite{Aref2018}. This means $\omega_1=-\omega_2>0$ if $\Delta t_2=-\Delta t_1>0$. This minimizes the overall pulse duration for all $z\in[0,L]$. 

We uniformly qunatize $\varphi_k$ to $m\cdot 2 \pi/128$, $m=1,\dots,128$. For each chosen combination of $\{\Delta t_k \coloneqq \Delta t_k(z=0)\}$ and $\{\omega_k\}$ we calculate $\widehat{T}$ and $\widehat{B}$ in \eqref{eq:TBP_def} as the maximum of $T$ and $B$ for the resulting $128^{N-1}$ phase combinations of $\{\varphi_k\}$.
For $N=2$, we find the minimum of $\widehat{T} \cdot \widehat{B}$ by exhaustive search for all combinations for $\omega_1$ and $\Delta t_1$, where we quantize $\omega_1$ in steps of $0.05$ in the range $[0,1]$ and $\Delta t_1$ in steps of $0.2$ in $[-4,0]$. For $N=3$ we search additionally the best combination of $\omega_3$ and $\Delta t_3$ quantized in $[-1,1]$ in steps of $0.1$ and $[-3,0]$ in steps of $0.2$ respectively. We increase the resolution around the observed optima to $0.01$ for $\omega_k$ and $0.05$ for $\Delta t_k$. Due to the precompensation, the optimization result is only valid for the specific (normalized) transmission distance that follows from the optimum parameters ${L^\star=\left|\frac{\Delta t_1^\star}{2\omega_1^\star}\right|}$. The respective optimization results $\{\Delta t_1^\star\}$ and $\{\omega_1^\star\}$ that minimize $\widehat{T} \cdot \widehat{B}$ are given in Tab.~\ref{tab:real_const_result} for $N=2,3$ and the two definitions (energy based and threshold based) for $T$ and $B$. 
\begin{figure}
\begin{floatrow}
\ffigbox{%
\begin{tikzpicture}[baseline=(current axis.south)]
\hspace{-20mm}
\begin{axis}[ticklabel style={font=\footnotesize},width=0.5\columnwidth,height=0.18\textheight,xmin=-10,xmax=10,ymin=0,ymax=1.5,y label style={yshift=-1em},label style={font=\small},xlabel={$t$},ylabel={$|q(t)|$},x label style={yshift=0.5em},legend style={font=\tiny},title style={font=\small}]     
 
    \addplot[forget plot,thick,myblue,width=\linewidth]table {2Sol_imag_optPulse1.dat};
    \addplot[forget plot,thick,myblue,width=\linewidth,dashed]table {2Sol_imag_optPulse2.dat};
    \addplot[forget plot,thick,myblue,width=\linewidth,dotted]table {2Sol_imag_optPulse3.dat};
    \addplot[forget plot,thick,myblue,width=\linewidth,dash dot]table {2Sol_imag_optPulse4.dat};
    
\end{axis}
\node at (0,-0.8) {(a)};
\hspace{35mm}
\begin{axis}[ticklabel style={font=\footnotesize},width=0.5\columnwidth,height=0.18\textheight,xmin=-15,xmax=15,ymin=0,ymax=1.5,y label style={yshift=-1em},label style={font=\small},xlabel={$t$},x label style={yshift=0.5em},legend style={font=\tiny,at={(1.2,-0.25)},overlay,legend columns=4},title style={font=\small}]     
 
    \addplot[forget plot,thick,myblue,width=\linewidth]table {3Sol_imag_optPulse1.dat};
    \addplot[forget plot,thick,myblue,width=\linewidth,dashed]table {3Sol_imag_optPulse2.dat};
    \addplot[forget plot,thick,myblue,width=\linewidth,dotted]table {3Sol_imag_optPulse3.dat};
    \addplot[forget plot,thick,myblue,width=\linewidth,dash dot]table {3Sol_imag_optPulse4.dat};

\end{axis}

\node at (0,-0.8) {(b)};
\end{tikzpicture}
}{%
  \caption{Optimum soliton pulses for (a) $N=2$ and (b) $N=3$, achieving the smallest $(\widehat{T} \cdot \widehat{B})^\star$ for four different phase combinations of $\{\varphi_k\}$.}\label{fig:optimum_pulse_2Sol_imag}
}
\ffigbox{%
\begin{tikzpicture}[baseline=(current axis.south)]
\hspace{-25mm}
\begin{axis}[ticklabel style={font=\footnotesize},width=0.225\textwidth,height=0.18\textheight,xmin=-10,xmax=10,ymin=0,ymax=2.5,y label style={yshift=-1.5em},label style={font=\small},ylabel={$|q(t)|$},xlabel={$t$},x label style={yshift=0.5em},legend style={font=\tiny},title style={font=\small}]     
 
    \foreach \k in {1,...,9} 
        {
    \addplot[forget plot,thick,myred,width=\linewidth]table {3Sol_real_opt_Tx\k.dat};

        }

\end{axis}
\node at (0,-0.8) {(a)};
\node at (0.55,2.1) {$z=0$};
\hspace{24mm}
\begin{axis}[ticklabel style={font=\footnotesize},width=0.225\textwidth,height=0.18\textheight,xmin=-10,xmax=10,ymin=0,ymax=2.5,label style={font=\small},xlabel={$t$},yticklabels=none,x label style={yshift=0.5em},legend style={font=\tiny},title style={font=\small}]       
    
    \foreach \k in {1,...,9} 
        {
    \addplot[forget plot,thick,myred,width=\linewidth]table {3Sol_real_opt_mid\k.dat};

        }

\end{axis}
\node at (0,-0.8) {(b)};
\node at (0.55,2.1) {$z=\frac{L}{2}$};
\hspace{24mm}
\begin{axis}[ticklabel style={font=\footnotesize},width=0.225\textwidth,height=0.18\textheight,xmin=-10,xmax=10,ymin=0,ymax=2.5,y label style={yshift=-1.5em},label style={font=\small},xlabel={$t$},yticklabels=none,x label style={yshift=0.5em},legend style={font=\tiny},title style={font=\small}]       
    
    \foreach \k in {1,...,9} 
        {
    \addplot[forget plot,thick,myred,width=\linewidth]table {3Sol_real_opt_Rx\k.dat};

        }

\end{axis}
\node at (0,-0.8) {(c)};
\node at (0.55,2.1) {$z=L$};
\end{tikzpicture}
}{%
\caption{Optimum third order soliton pulse for different phase combinations of $\{\varphi_k\}$ at (a) transmitter, (b) along the fiber at $z=L/2$ and (c) at the receiver.}
\label{fig:optimum_pulse_3Sol_real}
}
\end{floatrow}
\end{figure}
\begin{figure}
\hspace{-8mm}
\begin{tikzpicture}[baseline=(current axis.south)]
\begin{axis}[axis y line*=left,ymin=13, ymax=20,ticklabel style={font=\footnotesize},width=0.8\textwidth,height=0.21\textheight,xmin=0,xmax=6,y label style={yshift=-1em},label style={font=\small},ylabel={$T_{\rm max}$},xlabel={$z$},y label style={myred},x label style={yshift=0.5em},legend style={font=\small},yticklabel style={myred},title style={font=\small},y axis line style={myred}]

    \addplot[mark=square,mark repeat=5,forget plot,thick,myred,width=\linewidth]table {T_2Sol_real_energy.dat};

    \addplot[mark=o,mark repeat=4,forget plot,thick,myred,width=\linewidth]table {T_3Sol_real_energy-1.dat};

\end{axis}
\begin{axis}[axis y line*=right,ylabel={$B_{\rm max}$},ylabel near ticks, yticklabel pos=right,ymin=0.8, ymax=3.5,axis x line=none,ticklabel style={font=\footnotesize,myblue},width=0.8\textwidth,height=0.21\textheight,xmin=0,xmax=6,label style={font=\small,myblue},legend style={font=\small,at={(0.975,1.25)}},overlay,legend columns=4,title style={font=\small},y axis line style={myblue}]

    \addplot[mark=square*,mark repeat=4,forget plot,thick,myblue,width=\linewidth]table {B_2Sol_real_energy.dat};

    \addplot[mark=*,mark repeat=2,forget plot,thick,myblue,width=\linewidth]table {B_3Sol_real_energy-1.dat};

\addlegendimage{line legend,myred,mark=square}
\addlegendentry{$N=2$, $T_{\rm max}$}
\addlegendimage{line legend,myblue,mark=square*}
\addlegendentry{$N=2$, $B_{\rm max}$}
\addlegendimage{line legend,myred,mark=o}
\addlegendentry{$N=3$, $T_{\rm max}$}
\addlegendimage{line legend,myblue,mark=*}
\addlegendentry{$N=3$, $B_{\rm max}$}

\end{axis}
\end{tikzpicture}

\caption{Time-bandwidth product $T_\mathrm{max} B_\mathrm{max}$ during propagation for optimum second/third order soliton with $\lambda_k^*=\omega_k^*+j\sigma$.}
\label{fig:TBP_vs_z_2Sol_real_opt}
\end{figure}
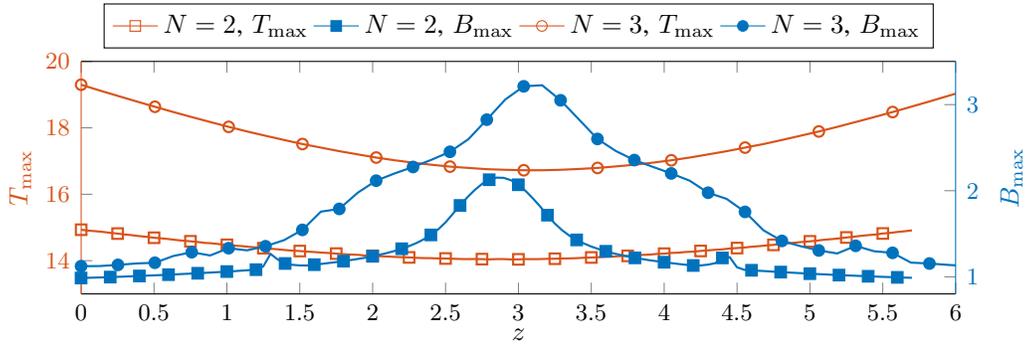

Note that multiple parameter sets achieve similar optimization results for $\widehat{T} \cdot \widehat{B}$. However the achieved $\widehat{T} \cdot \widehat{B}$ are very close. Thus the given parameter set $\{\omega_k^\star\}$ and $\{\Delta t_k^\star\}$ result in soliton pulses very close to the optimum. Using the resulting optimized third order pulses, an improvement in the time-bandwidth product per eingenvalue of $\overline{T \cdot B}_3 / \overline{T \cdot B}_1\approx 0.73$ can be achieved in this scenario. For $N=2$ we achieve $\overline{T \cdot B}_2 / \overline{T \cdot B}_1\approx 0.75$.
The resulting optimum soliton pulse for $N=3$ is shown in Fig.~\ref{fig:optimum_pulse_3Sol_real} for different phase combinations of $\{\varphi_k\}$ at transmitter $z=0$, along the fiber $z=L/2$ and at the receiver $z=L$. Note that, althoug changing during propagation, the $\Delta t_k$ will keep their respective relation. From \eqref{eq:Qd_evol}, it follows that after propagation along a distance ${L=\left|\frac{\Delta t_k(z=0)}{2\omega_k }\right|}$, the $\Delta t_k(z)$ are transformed to $\Delta t_k(z=L)=-\Delta t_k(z=0)$. Thus the resulting multi-solitons achieve their $\widehat{T}$ at $z=0$ as well as at $z=L$ and have the same $\widehat{B}$ at $z=0$ and at $z=L$. This scenario with eigenvalues parallel to the real axis outperforms the scenario with purely imaginary eigenvalues, however, only for the fixed normalized propagation distance that matches to the optimum soliton parameters as ${L^\star=\left|\frac{\Delta t_1^\star}{2\omega_1^\star}\right|}$. For the optimum parameters from Tab.~\ref{tab:real_const_result}, we get $L \approx 2$.
The intuitive reason is that due to changing $\Delta t_k$ along the propagation, we can avoid large bandwidths at transmitter and receiver by "shifting" the bandwidth expansion to the center of the transmission link. For propagation distances larger than $L^\star$, the solitons further separate and the overall pulse duration gets larger which reduces the achievable time-bandwidth product gain. We illustrate this influence of the changing $\Delta t_k$ on the pulse duration and bandwidth during propagation. We define \eqref{eq:TmaxBmax} where $T_\mathrm{max}$ and $B_\mathrm{max}$ are the maximum pulse duration and bandwidth that could be observed at a \emph{specific} $z$ along the fiber.
\begin{equation}\label{eq:TmaxBmax}
T_{\max}=\underset{\varphi_k, 1\leq k\leq N}{\max} T \qquad B_{\max}=\max_{\varphi_k, 1\leq k\leq N} B, 
\end{equation}
Comparing with \eqref{eq:TBP_def}, the time-bandwidth product is given by the maximization of $T_{\max}$ and $B_{\max}$ along $z=[0,L]$ or $z=\{0,L\}$, respectively. We show the evolution of $T_{\max}$ and $B_{\max}$ during propagation along $z\in [0,L]$ in Fig.~\ref{fig:TBP_vs_z_2Sol_real_opt} for the optimal $2-$ and $3-$solitons.
We observe these solitons indeed achieving their maximum $T_{\max}$ at $z=0$ as well as at $z=L$ leading to $\widehat{T}=\left.T_{\max}\right|_{z=0}$. Furthermore they have the same $B_{\max}$ at $z=0$ and at $z=L$ giving $\widehat{B}=\left.B_{\max}\right|_{z=0}$.

\section{Estimation of Time-Bandwidth Product}\label{sec:TBP_higherN}

From the optimization results in the previous section, we observe the improvement in the time-bandwidth product per eigenvalue being small and slowly decreasing in $N$. The complexity of the numerical brute-force optimization, however, grows exponentially in the number of eigenvalues. Instead of finding the optimal soliton pulses, we estimate the smallest possible $(\widehat{T}\cdot \widehat{B})^\star$, to see how the time-bandwidth product per eigenvalue could potentially be decreased for higher $N$.

Finding the joint minimization $\min \widehat{T}\cdot \widehat{B}$ in \eqref{eq:TBP_minimization} is hard. Therefore we look at\begin{equation}\label{eq:minTminB}
(\widehat{T}\cdot \widehat{B})^\star \geq \min_{\{\lambda_k\}_{k=1}^N} \left( \min_{\substack{\{\Delta t_k\}_{k=1}^N \\  L}} \widehat{T}  \min_{\substack{\{\Delta t_k\}_{k=1}^N \\  L}} \widehat{B} \right)
\end{equation}
where we minimize pulse duration and bandwidth individually with respect to $\Delta t_k$. For each minimization, the minima are achieved for zero propagation distance, $L=0$ (pulse duration and bandwidth change during propagation). Then using \eqref{eq:TBP_def} and the definition \eqref{eq:TmaxBmax}, we get 
\begin{equation}\label{eq:minTminB2}
(\widehat{T}\cdot \widehat{B})^\star \geq \min_{\{\lambda_k\}_{k=1}^N} \left( \min_{\{\Delta t_k\}_{k=1}^N } T_{\max} \min_{\{\Delta t_k\}_{k=1}^N } B_{\max} \right)
\end{equation}

We first ask how pulse duration $T_{\max}$ and bandwidth $B_{\max}$ change in terms of $\{\Delta t_k\}$. Fig.~\ref{fig:T_B_vs_ln(eta)} illustrates $T_\mathrm{max}$ and $B_\mathrm{max}$ in terms of $\Delta t_2$ for two different exemplary $2-$solitons. We observe that the smallest $T_\mathrm{max}$ is attained at $\Delta t_2=0$, whereas $B_\mathrm{max}$ reaches its maximum. On the other hand, $B_\mathrm{max}$ reaches its minimum for $\Delta t_2 \to \infty$, for which $T_\mathrm{max}$ grows unboundedly. 
We observe this behavior for different eigenvalues and soliton orders $N$. The intuitive reason is, that each multi-soliton splits into its first order components when $|\Delta t_k - \Delta t_l| \to \infty \quad \forall k>l$; meaning $N$ seperate 1-solitons without any interaction. As $|\Delta t_k - \Delta t_l|$ decreases, the distance between these 1-solitons decreases. This reduces the overall pulse duration $T_\mathrm{max}$, however the bandwidth $B_\mathrm{max}$ increases due to higher nonlinear interaction. We use these observations to estimate  
\begin{align}\label{eq:TB_limitCases}
\min_{\{\Delta t_k\}_{k=1}^N } T_{\max} \qquad & \mathrm{when} \qquad \Delta t_k = \Delta t \quad \forall k=1,\dots,N	\nonumber
\\
\min_{\{\Delta t_k\}_{k=1}^N } B_{\max} \qquad & \mathrm{when} \qquad \Delta t_k - \Delta t_l \to \infty \quad \forall k>l
\end{align}

In the following, we give analytical approximations for $T_\mathrm{max}$ and $B_\mathrm{max}$ in these respective limit cases. Based on the Darboux transform Alg.~\ref{alg:DT2}, we derive the following approximations on pulse duration given in Thm.~\ref{th:T_imag} and Thm.~\ref{th:T_real}. For detailed derivations see App.~\ref{sec:proof_Th12}.

\begin{theorem}\label{th:T_imag}
Consider an $N-$soliton with eigenvalues $\lambda_k=j\sigma_k$, ${\sigma_k \in \mathbb{R}^+}$, and spectral amplitude scaling $\eta_k=\exp(2\sigma_k \Delta t)$  (${\Delta t_k = \Delta t}$), $1\leq k \leq N$. We assume (w.l.o.g) $\sigma_N=\underset{k}{\min}\,{\sigma_k}$. This soliton's pulse duration $T_\mathrm{max}(\varepsilon)$, according to Def.~\ref{def:TB_def_energy}, can be well approximated by 
\begin{align}\label{eq:Tmin_imag_const}
 T_{\max}(\varepsilon) \approx T_\mathrm{lim,im} =
 \frac{1}{2\sigma_N}\left( \ln\left(\frac{2}{\varepsilon}\frac{\sigma_N}{\sum_{k=1}^{N}\sigma_k}\right) + 2\sum_{k=1}^{N-1} \ln \left(\left|\frac{\sigma_N +\sigma_k}{\sigma_N-\sigma_k}\right|\right)\right)
\end{align}
The approximation becomes tight for $\varepsilon \to 0$. 
\end{theorem}

\begin{theorem}\label{th:T_real}
Consider an $N-$soliton with eigenvalues $\lambda_k=j\sigma+\omega_k$ parallel to the real axis ($\sigma \in \mathbb{R}^+, \omega_k \in \mathbb{R}$) and spectral amplitude scaling $\eta_k$. We denote $A_r=\left|\sum_{k=r}^{N}  a_{r,k}\right|$ where $a_{r,k}$ are calculated from ~\eqref{eq:rho_update} in Alg.~\ref{alg:DT2} such that $\rho_k^{(k-1)}(t)=\sum_{r=1}^k a_{r,k} \rho_r^{(0)}(t)$. This soliton's pulse duration $T_\mathrm{max}(\varepsilon)$, according to Def.~\ref{def:TB_def_energy}, can be well approximated by \eqref{eq:T_real}. The minimization of this approximation is attained when $\eta_k=\eta, \quad \forall 1\leq k \leq N$ and we denote $T_\mathrm{lim,re}=\underset{\eta_k}{\min}\, T_\mathrm{max}(\varepsilon)$. 
\begin{equation}\label{eq:T_real}
T_\mathrm{max}(\varepsilon) \approx \frac{1}{2\sigma}\ln\left( \frac{2}{N \varepsilon} \left(\sum_{r=1}^{N} \eta_r A_r \right) \left(\sum_{r=1}^{N} \frac{1}{\eta_r} A_r \right)\right)
\end{equation}

\end{theorem}

As \emph{numerically} observed, the minimum of $B_\mathrm{max}$ is attained for multi-solitons being separated into their first order components \eqref{eq:TB_limitCases} where $\eta_k/\eta_l \to \infty \quad \forall k > l$ (which is equivalent to $\Delta t_k - \Delta t_l \to \infty \quad \forall k>l$). In this case, the pulse is a linear superposition of fundamental solitons. From their available analytical description, we derive the following approximations on bandwidth given in Thm.~\ref{th:B_imag} and Thm.~\ref{th:B_real}. For details see Appx.~\ref{sec:proof_Th34}.

\begin{theorem}\label{th:B_imag}
Consider an $N-$soliton with imaginary eigenvalues $\lambda_k=j\sigma_k$, $\sigma \in \mathbb{R}^+$, separated into its first order components where $\eta_k/\eta_l \to \infty \quad \forall k > l$ and $\sigma_1=\underset{k}{\max} \, \sigma_k$. Its bandwidth $B_{\max}$, according to Def.~\ref{def:TB_def_energy}, can be well approximated by 
\begin{equation}\label{eq:Bmin_real_const_energy}
B_{\max}(\varepsilon) \approx B_\mathrm{lim,im} = \frac{2 \sigma_1}{\pi^2} \ln\left(\frac{2}{\varepsilon}\frac{\sigma_1}{\sum_{k=1}^{N}\sigma_k}\right)
\end{equation}
\end{theorem}

\begin{theorem}\label{th:B_real}
Consider an $N-$soliton with eigenvalues ${\lambda_k=j\sigma+\omega_k}$ parallel to the real axis ($\sigma \in \mathbb{R}^+, \omega_k \in \mathbb{R}$) separated into its first order components where $\eta_k/\eta_l \to \infty, \quad \forall k > l$. Its bandwidth $B_{\max}$, according to Def.~\ref{def:TB_def_energy}, can be well approximated by 
\begin{align}\label{eq:Bmin_real_const_energy}
 B_{\max}(\varepsilon) \approx	B_\mathrm{lim,re} = 
 \frac{2\sigma}{\pi^2} \left( \ln\left(\frac{2}{\varepsilon N}\right) + \ln \left( \sum_{k=1}^N \exp(\frac{\pi \omega_k}{2 \sigma})\cdot \sum_{k=1}^N \exp(-\frac{\pi \omega_k}{2 \sigma})\right) \right)
\end{align}
\end{theorem}

Note that the two limiting cases of minimal $T_\mathrm{max}$ and minimal $B_\mathrm{max}$ can not be attained simultaneously (compare Fig.~\ref{fig:T_B_vs_ln(eta)}). However, ${\underset{\Delta t_k}{\min}\, T_\mathrm{max}\cdot \underset{\Delta t_k}{\min}\, B_\mathrm{max}}$ in \eqref{eq:minTminB2} is a lower bound on the achievable $T_\mathrm{max}B_\mathrm{max}$ for which we can given an analytical expression using the approximations in Thm.~\ref{th:T_imag},\ref{th:T_real},\ref{th:B_imag},\ref{th:B_real}. We approximate ${\underset{\Delta t_k}{\min}\, T_\mathrm{max}\cdot \underset{\Delta t_k}{\min}\, B_\mathrm{max}}$ as $T_\mathrm{lim,im/re} \cdot B_\mathrm{lim,im/re}$ and thus \eqref{eq:minTminB2} as
\begin{equation}
\left(\widehat{T}\cdot \widehat{B}\right)^\star  \gtrapprox \min_{\{\lambda_k\}_{k=1}^N} \left(T_\mathrm{lim,im/re} \cdot B_\mathrm{lim,im/re} \right).
\end{equation}

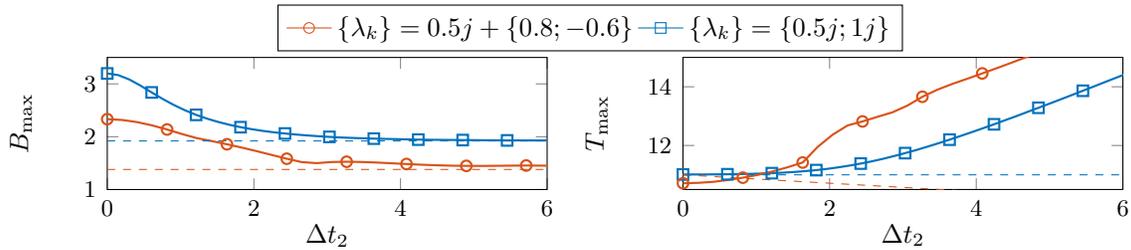
\begin{figure}
\vspace{2mm}
\hspace{-4mm}
\begin{tikzpicture}[baseline=(current axis.south)]
\begin{axis}[ticklabel style={font=\footnotesize},width=0.45\textwidth,height=0.15\textheight,xmin=0,xmax=6,ymin=1,ymax=3.5,y label style={yshift=-0.2em},label style={font=\small},ylabel={$B_\mathrm{max}$},xlabel={$\Delta t_2$},x label style={yshift=0.5em},legend style={font=\tiny},title style={font=\small}]     

    \addplot[forget plot,thick,myblue,mark=square,mark repeat=4,width=\linewidth]table {B_vs_DeltaT_imagConst.dat};

    \addplot[forget plot,thick,myred,mark=o,mark repeat=4,width=\linewidth]table {B_vs_DeltaT_realConst.dat};
    
    \addplot[forget plot,mark=none, myblue, samples=100, dashed] coordinates {(0,1.92) (7,1.92)};
    
    \addplot[forget plot,mark=none, myred, samples=100, dashed] coordinates {(0,1.38) (7,1.38)};

\end{axis}
\end{tikzpicture}
\begin{tikzpicture}[baseline=(current axis.south)]
\begin{axis}[ticklabel style={font=\footnotesize},width=0.45\textwidth,height=0.15\textheight,xmin=0,xmax=6,ymin=10.5,ymax=15,y label style={yshift=-0.2em},label style={font=\small},ylabel={$T_\mathrm{max}$},xlabel={$\Delta t_2$},x label style={yshift=0.5em},legend style={font=\small,at={(0.5,1.4)},overlay,legend columns=2},title style={font=\small}]     

    \addplot[forget plot,thick,myblue,mark=square,mark repeat=4,,width=\linewidth]table {T_vs_DeltaT_imagConst.dat};

    \addplot[forget plot,thick,myred,mark=o,mark repeat=4,width=\linewidth]table {T_vs_DeltaT_realConst.dat};

\addplot[forget plot,mark=none, myblue, samples=100, dashed] coordinates {(0,11) (7,11)};
\addplot[forget plot,mark=none, myred, samples=100, dashed] coordinates {(0,11) (7,10)};
     
\addlegendimage{line legend,mark=o,myred}
\addlegendentry{$\{\lambda_k\}=0.5j+\{0.8;-0.6\}$}
\addlegendimage{line legend,mark=square,myblue}
\addlegendentry{$\{\lambda_k\}=\{0.5j;1j\}$}

\end{axis}
\end{tikzpicture}
\caption{Pulse duration $T_\mathrm{max}$ and bandwidth $B_\mathrm{max}$ for exemplary second order solitons as a function of $\Delta t_2$ with $\Delta t_1=0$.}
\label{fig:T_B_vs_ln(eta)}
\end{figure}

For the time-bandwidth product per eigenvalue, following definition \eqref{eq:TBP_per_eigenvalue}, we have
\begin{equation}\label{eq:TBPmin_approx}
\overline{T \cdot B}_N \gtrapprox \frac{1}{N} \min_{\{\lambda_k\}_{k=1}^N} \left( T_\mathrm{lim,im/re} \cdot B_\mathrm{lim,im/re} \right).
\end{equation}
The above lower bound \eqref{eq:TBPmin_approx} can be minimized numerically. We normalize again by $\overline{T \cdot B}_1$ to indicate the gain compared to a first order pulse. In Fig.~\ref{fig:TBP_LB}, the blue "$\square$" line shows the minimization result for imaginary eigenvalues while the red "$\circ$" line corresponds to the case with eigenvalues parallel to the real axis.
The brute-force optimized $\overline{T \cdot B}_N / \overline{T \cdot B}_1$ in Sec.~\ref{sec:num_pulse_opt} are marked correspondingly in Fig.~\ref{fig:TBP_LB}. Non-filled marks indicate the optimization result for defining pulse duration and bandwidth according to Def.~\ref{def:TB_def_energy}, whereas filled marks correspond to Def.\ref{def:TB_def_threshold}. The analytical result indicates that the time-bandwidth product per eigenvalue can be decreased for higher soliton orders, however the improvement is small and only slowly growing in $N$. We observe the same behavior for the optimum soliton pulses obtained from brute-force optimizaiton. However, the gap between analytical approximation and practically achieved values seems to become larger for higher $N$. One can observe the eigenvalue constellation with non-zero real part potentially achieving a higher gain.
However, to be competitive to classical Nyquist-based transmission schemes, multi-solitons should achieve a time-bandwidth product in the range of $1$ (time-bandwidth product of a classical Nyquist pulse). But this target value, indicated by the black dashed line in Fig.~\ref{fig:TBP_LB}, appears to have a large gap to the values achievable by the soliton pulses.

\begin{figure}
\centering
\vspace{5mm}
\begin{tikzpicture}[baseline=(current axis.south)]
\begin{axis}[ticklabel style={font=\footnotesize},width=0.8\textwidth,height=0.2\textheight,xmin=1,xmax=10,ymin=0,ymax=1,y label style={yshift=0em},label style={font=\small},ylabel={$\overline{TB}_N/\overline{TB}_1$},xlabel={$N$},x label style={yshift=0.5em},legend style={font=\small,at={(1,1.39)},overlay,legend columns=3},title style={font=\small}]

    \addplot[forget plot,thick,myred,mark=o,width=2\linewidth]table {TBP_per_EV_LB_real_const_energy.dat};

     \addplot[forget plot,thick,myblue,mark=square,width=2\linewidth]table {TBP_per_eigenvalue_LB_imag_energy.dat};

    \addplot[forget plot,mark=none, black, samples=100, dashed] coordinates {(0,0.1) (20,0.1)};
    
    \addplot[mark=square,myblue,forget plot] coordinates {(2,0.89)};	 
    \addplot[mark=o,myred,forget plot] coordinates {(2,0.74)};
    \addplot[mark=square,myblue,forget plot] coordinates {(3,0.84)};	     	
    \addplot[mark=o,myred,forget plot] coordinates {(3,0.71)};  
                    
    \addplot[mark=square*,thick,myblue,forget plot] coordinates {(2,0.86)};
    \addplot[mark=square*,thick,myblue,forget plot] coordinates {(3,0.855)};
    \addplot[mark=*,myred,forget plot] coordinates {(2,0.75)};        
    \addplot[mark=*,thick,myred,forget plot] coordinates {(3,0.72)};

\addlegendimage{line legend,mark=square,myblue}
\addlegendentry{$\lambda_k=j\sigma_k$, energy}
\addlegendimage{only marks,myblue,mark=square}
\addlegendentry{$\lambda_k=j\sigma_k$, energy}
\addlegendimage{only marks, mark=square*,myblue}
\addlegendentry{$\lambda_k=j\sigma_k$, threshold}

\addlegendimage{line legend,mark=o,myred}
\addlegendentry{$\lambda_k=j\sigma+\omega_k$, energy}
\addlegendimage{only marks,mark=o,myred}
\addlegendentry{$\lambda_k=j\sigma+\omega_k$, energy}
\addlegendimage{only marks,mark=*,myred}
\addlegendentry{$\lambda_k=j\sigma+\omega_k$, threshold}

\end{axis}
\end{tikzpicture}

\caption{Achievable gain in time-bandwidth product per eigenvalue from numerical pulse optimization for $N=2,3$: $T$ and $B$ defined according to Def.~\ref{def:TB_def_energy} (non-filled marks) or Def.~\ref{def:TB_def_threshold} (filled marks) with $\varepsilon=10^{-4}$. Lines show the analytical lower bound estimation for different soliton orders $N$. The black dashed line indicates the necessary gain to achieve the time-bandwidth product of a Nyquist pulse. }
\label{fig:TBP_LB}
\end{figure}
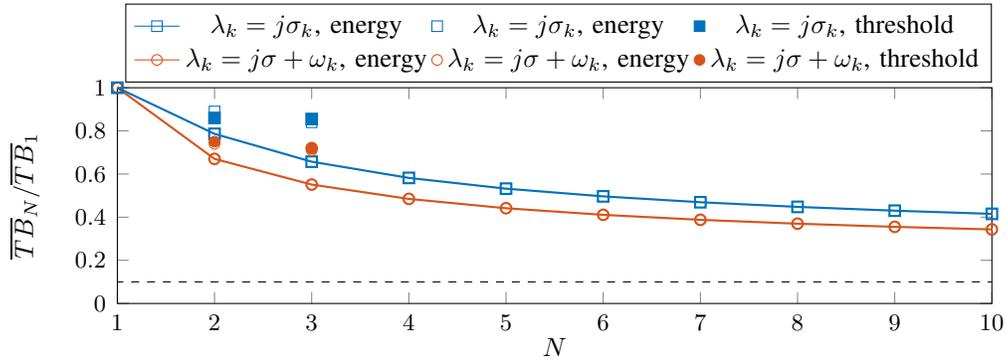

\section{Conclusion}\label{sec:conclusion}

We studied the pulse duration and bandwidth of multi-solitons for a transmission scenario where the phases of spectral amplitudes are modulated. We define the time-bandwidth product such that it takes into account the pulse variations due to modulation and propagation. For second and third order solitons, we numerically optimized the location of eigenvalues and magnitudes of spectral amplitudes in order to minimize the time-bandwidth product. The optimized multi-soliton pulses turn out to look similar to a train of first order pulses. This shape reduces the nonlinear interaction and thus the bandwidth expansion. These optimum pulses decrease the time-bandwidth product per degree of freedom (eigenvalue). However the improvement is slow in the soliton order $N$ and quite small.
For arbitrary solitons, we generally identify two limit cases where either pulse duration or bandwidth becomes minimal for given eigenvalues. We give analytical approximations on pulse duration and bandwidth in these limit cases and numerically minimize them in terms of eigenvalues. The result is an approximation on the smallest possible time-bandwidth product per degree of freedom as a function of the soliton order $N$. One should note that the maximization of pulse duration and bandwidth over all phase combinations of spectral amplitudes is a worst case assumption that simplifies the optimization and clearly increases the time-bandwidth product result. Better results may be achievable by, e.g., considering only a finite number of phases and tuning magnitudes of spectral amplitudes accordingly.

\appendices
\section{Proof of Theorem \ref{th:SolProperties}}\label{sec:sol_prop_proof}

Properties (i) to (iv) are already shown in \cite{yousefi2014nft} using the Zakharov-Shabat system \eqref{eq:ZS}. Property (v*) was partially shown in \cite{Haus1985}. In more details, it is shown in \cite{Span2017} that (v) is a necessary and sufficient condition for the pulse symmetry.
Here, we prove all properties using the recursive Darboux transform in Alg.~\ref{alg:DT2}. Using \eqref{eq:Qd_eta}, \eqref{eq:Qd_0}, the initialization in \eqref{eq:Darboux_init} is \eqref{eq:rho_init_etak} with the resulting time domain signal from \eqref{eq:sig_update} given as \eqref{eq:sig_rho}.
\begin{align}
\rho_k^{(0)}(t) & =\eta_k \exp(j\varphi_k) \exp(2j\omega_k t) \exp(-2\sigma_k t).\label{eq:rho_init_etak}
\\
q^{(N)}(t) & = -4 \sum_{k=1}^{N} \sigma_k \frac{\rho_k^{*(k-1)}(t)}{1+\left|\rho_k^{(k-1)}(t)\right|^2}.\label{eq:sig_rho}
\end{align}

Let $\tilde{\rho}_k^{(i)}(t)$ and $\tilde{q}^{(N)}(t)$ denote the transformed $\rho_k^{(i)}(t)$ and, respectively, the transformed $q^{(N)}(t)$, if we apply one of the properties (i) to (vi). First we prove properties (i) to (iv) and (vi):

Applying one of the properties (i) to (iv), the transformed Darboux initializations \eqref{eq:Darboux_init} are
\begin{flalign}
\mathrm{(i)} \qquad \tilde{\rho}_k^{(0)}(t) & = \eta_k \exp(2j\omega_k t) \exp(j\varphi_k-j\varphi_0) \exp(-2\sigma_k t)	= \exp(-j\varphi_0)\rho_k^{(0)}(t)	\label{eq:prop1_init}
\\
\mathrm{(ii)} \qquad \tilde{\rho}_k^{(0)}(t) & = \eta_k \exp(2\sigma_k t_0) 
\exp(j\varphi_k-j2\omega_k t_0+2j\omega_k t) \exp(-2\sigma_k t)	
 = \rho_k^{(0)}(t-t_0)	\label{eq:prop2_init}
\\
\mathrm{(iii)} \qquad \tilde{\rho}_k^{(0)}(t) & =\eta_k \exp(j\varphi_k) \exp(2j\frac{\omega_k}{\sigma_0} t) \exp(-2\frac{\sigma_k}{\sigma_0} t)	
 =\rho_k^{(0)}\left(\frac{t}{\sigma_0}\right)	\label{eq:prop3_init}
\\
\mathrm{(iv)} \qquad \tilde{\rho}_k^{(0)}(t) & =\eta_k \exp(j\varphi_k) \exp(2j(\omega_k-\omega_0) t) \exp(-2 \sigma_k t)	
 =\exp(-2j\omega_0 t)\rho_k^{(0)}(t)	\label{eq:prop4_init}
\\
\mathrm{(vi)} \qquad \tilde{\rho}_k^{(0)}(t) & = \eta_k \exp(-2j\omega_k t) \exp(-j\varphi_k) \exp(-2\sigma_k t)	
 = \rho_k^{*\, (0)}(t)	\label{eq:prop6_init}
\end{flalign}

For the update equation \eqref{eq:rho_update}, it is easy to check that if one of the above relations holds between $\tilde{\rho}_k^{(j)}(t)$ and $\rho_k^{(j)}$ for all ${k=j+1,\dots,N}$, the same relation will hold between $\tilde{\rho}_k^{(j+1)}$ and $\rho_k^{(j+1)}$ for all $k>j+1$. Since the respective relation holds at $j=0$, we conclude that the relation is preserved for any $j=1,\dots,N-1$. Thus, all the above properties hold in particular between $\tilde{\rho}_k^{(k-1)}$ and $\rho_k^{(k-1)}$ for all $k=1,\dots,N$.
Following \eqref{eq:sig_rho}, the transformed signal $\tilde{q}^{(N)}(t)$ is
\begin{equation}\label{eq:sig_rho_tilde}
\tilde{q}^{(N)}(t) = -4 \sum_{k=1}^{N} \sigma_k \frac{\tilde{\rho}_k^{*(k-1)}(t)}{1+\left|\tilde{\rho}_k^{(k-1)}(t)\right|^2}.
\end{equation}
Replacing $\tilde{\rho}_k^{(k-1)}$ in \eqref{eq:sig_rho_tilde} by either of the above (i), (ii), (iii), (iv), (vi), 
one concludes
\begin{multicols}{2}
\begin{itemize}
\item[(i)] Given \eqref{eq:prop1_init}, $\quad \tilde{q}^{(N)}(t) = \exp(j\varphi_0) q^{(N)}(t)$

\item[(ii)] Given \eqref{eq:prop2_init}, $\tilde{q}^{(N)}(t) = q^{(N)}(t-t_0)$

\item[(iii)] Given \eqref{eq:prop3_init}, $\tilde{q}^{(N)}(t) = \frac{1}{\sigma_0} q^{(N)}(\frac{t}{\sigma_0})$

\item[(iv)] Given \eqref{eq:prop4_init}, $\tilde{q}^{(N)}(t) = \exp(2j\omega_0 t) q^{(N)}(t)$

\item[(iv)] Given \eqref{eq:prop4_init}, $\tilde{q}^{(N)}(t) = q^{*\, (N)}(t)$
\end{itemize}
\end{multicols}

Now we prove property (v): By replacing $\{\eta_k\}\rightarrow \{1/\eta_k\}$ and $\{\omega_k\}\rightarrow \{-\omega_k\}$ , the transformed Darboux initialization \eqref{eq:Darboux_init} writes as
\begin{flalign}\label{eq:rho_init_etak_inv}
\mathrm{(v)} \qquad \tilde{\rho}_k^{(0)}(t) = & \frac{1}{\eta_k} \exp(j\varphi_k) \exp(-2j\omega_k t) \exp(-2\sigma_k t)
=\frac{1}{\rho_k^{*(0)}(-t)}. &
\end{flalign}
In addition, the update rule \eqref{eq:rho_update} taking into account the change of eigenvalues ${\omega_k \to -\omega_k}$, or equivalently ${\lambda_k \to -\lambda_k^*}$, rewrites as
\begin{align}\label{eq:rho_update_etak}
\tilde{\rho}_k^{(j+1)}(t) = 
\frac{(\lambda_{j+1}^* - \lambda_k^*)\tilde{\rho}_{k}^{(j)}(t)  +\frac{\lambda_{j+1} - \lambda_{j+1}^*}{1+|\tilde{\rho}_{j+1}^{(j)}(t)|^2}(\tilde{\rho}_{k}^{(j)}(t)-\tilde{\rho}_{j+1}^{(j)}(t))}
{\lambda_{j+1} - \lambda_k^*-\frac{\lambda_{j+1} - \lambda_{j+1}^*}{1+|\tilde{\rho}_{j+1}^{(j)}(t)|^2}\left(1 + \tilde{\rho}_{j+1}^{*(j)}(t)\tilde{\rho}_k^{(j)}(t) \right)}.
\end{align} 

If the relation $\tilde{\rho}_k^{(j)}(t) \rho_k^{*(j)}(-t)=1$ holds for all ${k=j+1,\dots,N}$, it will hold for $j+1$ as well. That means $\tilde{\rho}_k^{(j+1)}(t) \rho_k^{*(j+1)}(-t)=1$ for all $k>j+1$. Let us verify, that the update rule \eqref{eq:rho_update_etak} indeed preserves this property. Given the property \eqref{eq:rho_init_etak_inv} for all ${k=j+1,\dots,N}$ at an arbitrary $j=0,\dots,N-1$, i.e. $\tilde{\rho}_k^{(j)}(t)=\frac{1}{\rho_k^{*(j)}(-t)}$, the update rule \eqref{eq:rho_update_etak} can be rewritten as
\begin{align}\label{eq:prop5_induction1}
 \tilde{\rho}_k^{(j+1)}(t) =
 \frac{(\lambda_{j+1}^* - \lambda_k^*)\frac{1}{\rho_{k}^{*(j)}(-t)}  +\frac{\lambda_{j+1} - \lambda_{j+1}^*}{1+\frac{1}{|\rho_{j+1}^{(j)}(-t)|^2}}\left(\frac{1}{\rho_{k}^{*(j)}(-t)}-\frac{1}{\rho_{j+1}^{*(j)}(-t)}\right)}
{\lambda_{j+1} - \lambda_k^*-\frac{\lambda_{j+1} - \lambda_{j+1}^*}{1+\frac{1}{|\rho_{j+1}^{(j)}(-t)|^2}}\left(1 + \frac{1}{\rho_{j+1}^{(j)}(-t)\rho_k^{*(j)}(-t)} \right)}.
\end{align} 

To prove the induction step, let us assume the induction statement is correct. Then the update rule should preserve property \eqref{eq:rho_init_etak_inv}. That means
\begin{align}
 \tilde{\rho}_k^{(j+1)}(t) & = \frac{1}{\rho_k^{*(j+1)}(-t)} \label{eq:prop5}
\\ 
& = \frac{(\lambda_{k}^* - \lambda_{j+1})-\frac{\lambda_{j+1}^* - \lambda_{j+1}}{1+|\rho_{j+1}^{(j)}(-t)|^2} (1+\rho_{j+1}^{(j)}(-t) \rho_k^{*(j)}(-t))}{(\lambda_{k}^* - \lambda_{j+1}^*)\rho_k^{*(j)}(-t) + \frac{\lambda_{j+1}^* - \lambda_{j+1}}{1+|\rho_{j+1}^{(j)}(-t)|^2} \left(\rho_{k}^{*(j)}(-t) - \rho_{j+1}^{*(j)}(-t) \right)}
\label{eq:prop5_induction2}
\end{align}

One can now verify, that \eqref{eq:prop5_induction1} and \eqref{eq:prop5_induction2} are identical. The time domain signal according to \eqref{eq:sig_rho} for the transformation  \eqref{eq:rho_init_etak_inv} is given by
\begin{equation}\label{eq:signal_prop5}
\tilde{q}^{(N)}(t) = -4 \sum_{k=1}^{N} \sigma_k \frac{\tilde{\rho}_k^{*(k-1)}(t)}{1+\left|\tilde{\rho}_k^{(k-1)}(t)\right|^2}.
\end{equation}

We have shown that property \eqref{eq:prop5} is valid for all $k>j+1$ and therefore as well for $\tilde{\rho}_k^{(k-1)}(t)$. Using this property in \eqref{eq:signal_prop5}, the transformed signal rewrites as
\begin{equation}\label{eq:rho_rho_tilde_symmetry}
\tilde{q}^{(N)}(t) = -4 \sum_{k=1}^{N} \sigma_k \frac{\rho_k^{*(k-1)}(-t)}{1+\left|\rho_k^{(k-1)}(-t)\right|^2}.
\end{equation}

Comparing \eqref{eq:rho_rho_tilde_symmetry} with the original signal \eqref{eq:sig_rho}, one finds
\begin{equation}\label{eq:q_symmetry}
\tilde{q}^{(N)}(t)= q^{(N)}(-t)
\end{equation}

Now we prove the special case (v*) when $\eta_k=1$ and $\omega_k=0$. From \eqref{eq:rho_init_etak_inv}, we have $\tilde{\rho}_k^{(0)}(t)=\rho_k^{(0)}(t)$ from which we can immediately conclude $\tilde{q}^{(N)}(t)=q^{(N)}(t)$.
But from \eqref{eq:q_symmetry}, we know $\tilde{q}^{(N)}(t)= q^{(N)}(-t)$. Thus, one can conclude $q^{(N)}(t)=q^{(N)}(-t)$.

\section{Proof of approximations Thm.~\ref{th:T_imag} and Thm.~\ref{th:T_real}}\label{sec:proof_Th12}

To approximate the soliton pulse duration, we derive analytical expressions for the pulse in its tails $t \to \pm \infty$. W.l.o.g. we assume the eigenvalues $\lambda_k=\omega_k+j\sigma_k$ ($\sigma_k \in \mathbb{R}^+$, $\omega_k \in \mathbb{R}$) being order as follows: $\sigma_1\geq \sigma_2 \geq ... \geq \sigma_N$

\begin{lemma}\label{lem:Darboux_iteration_limit}
Consider the limits $t \to \pm \infty$.\\ 
a) \quad For arbitrary $\lambda_k=\omega_k+j\sigma_k$, the update rule \eqref{eq:rho_update} in the Darboux algorithm \ref{alg:DT2} simplifies to
\begin{align} 
\rho_k^{(j+1)}(t) & \to \rho_k^{(j)}(t)\frac{\lambda_k-\lambda_{j+1}^*}{\lambda_k-\lambda_{j+1}}-\frac{2j\sigma_{j+1}}{\lambda_k-\lambda_{j+1}}\rho_{j+1}^{(j)}(t), \qquad t \to \infty	\label{eq:rho_update_t_infty_general}
\\
\frac{1}{\rho_k^{(j+1)}(t)} & \to \frac{1}{\rho_k^{(j)}(t)}\frac{\lambda_k-\lambda_{j+1}^*}{\lambda_k-\lambda_{j+1}}-\frac{2j\sigma_{j+1}}{\lambda_k-\lambda_{j+1}}\frac{1}{\rho_{j+1}^{(j)}(t)}, \qquad t \to - \infty.
\label{eq:rho_update_t_-infty_general}
\end{align}
b) \quad Thus, each $\rho_k^{(j)}(t) \propto \exp(-2\sigma_k t)$ preserves its exponential order for all iterations 

\end{lemma}

\begin{proof} Assume for some $j=0,\dots,N-1$, we have 
\begin{equation}\label{eq:rho_prop}
\rho_k^{(j)}(t) \propto \exp(-2\sigma_k t) \qquad \forall k=j+1, \dots, N
\end{equation}
which is true for $j=0$, as the Darboux initialization \eqref{eq:Darboux_init} is
$\rho_k^{(0)}= \eta_k \exp(j\varphi_k)\cdot \exp(2j\omega_k t)\cdot \exp(-2\sigma_k t)$.
Then one can easily verify \eqref{eq:rho_update_t_infty_general} by neglecting all terms\\ ${\rho_k^{(j)}(t)\cdot \rho_l^{(j)}(t) \propto \exp(-2(\sigma_k+\sigma_l) t) \to 0}$ ${\forall k,l=j+1, \dots, N}$ in \eqref{eq:rho_update} for the limit $t\to + \infty$. Similarly, for $t \to - \infty$, considering only the dominant terms $\rho_k^{(j)}(t) \propto \exp(-2 \sigma_k t) \to \infty$ and ${\rho_k^{(j)}(t)\cdot \rho_l^{(j)}(t) \propto \exp(-2(\sigma_k+\sigma_l) t) \to \infty}$ $\forall k,j=j+1, \dots, N$ in \eqref{eq:rho_update} leads to \eqref{eq:rho_update_t_-infty_general}. Recalling the ordering $\sigma_1\geq \sigma_2 ... \geq \sigma_N$ and $k>j+1$, Lem.~\ref{lem:Darboux_iteration_limit} b) follows from  \eqref{eq:rho_update_t_infty_general} and \eqref{eq:rho_update_t_-infty_general}.
\end{proof}

The update equation \eqref{eq:rho_update_t_infty_general} is a linear combination of $\rho_k^{(j)}(t)$ and $\rho_{j+1}^{(j)}(t)$. Similarly, from update equation \eqref{eq:rho_update_t_-infty_general}, $\frac{1}{\rho_k^{(j+1)}(t)}$ is a linear combination of $\frac{1}{\rho_k^{(j)}(t)}$ and $\frac{1}{\rho_{j+1}^{(j)}(t)}$. Therefore, 
\begin{equation}\label{eq:a_rk_superpos}
\rho_k^{(k-1)}(t)=\sum_{r=1}^k a_{r,k} \rho_r^{(0)}(t)
\qquad \qquad
\frac{1}{\rho_k^{(k-1)}(t)}=\sum_{r=1}^k \tilde{a}_{r,k} \frac{1}{\rho_r^{(0)}(t)}.
\end{equation}
The values of $a_{r,k}$ and $\tilde{a}_{r,k}$ are finite and depend on the eigenvalues. $a_{r,k}$ and $\tilde{a}_{r,k}$ can recursively be computed using the update equations \eqref{eq:rho_update_t_infty_general} and \eqref{eq:rho_update_t_-infty_general}. Since the update coefficients in \eqref{eq:rho_update_t_infty_general} and \eqref{eq:rho_update_t_-infty_general} are the same, one can expect that $\tilde{a}_{r,k}=a_{r,k}$. Now, we express the tails of an $N$-soliton in terms of $a_{r,k}$. We denote the envelope of $|x(t)|$ or its spectrum $|X(f)|$ as $\widehat{|x(t)|}$ and $\widehat{|X(f)|}$.  

\begin{lemma}\label{lem:sig_tails_gen}
Consider a multi-soliton pulse with $N$ eigenvalues $\lambda_k=\omega_k+j\sigma_k$ ($\sigma_k \in \mathbb{R}^+$, $\omega_k \in \mathbb{R}$) and spectral amplitude scaling $\eta_k$. $a_{r,k}$ are the scalars according to \eqref{eq:a_rk_superpos}.

a) As $t \to \pm \infty$, the tails of the multi soliton pulse tend to
\begin{align}\label{eq:q_abs_gen_limit}
 |q^{(N)}(t)| \underset{t \to \pm \infty}{\to} 
\left|4 \sum_{r=1}^{N} \eta_r^{\pm 1} \exp(\mp 2\sigma_r t) \exp(\pm j\varphi_r \pm 2j\omega_r t) \sum_{k=r}^{N} \sigma_k a_{r,k}\right|.
\end{align}

b) The signal envelope in the tails $t \to \pm \infty$ is given as
\begin{align}\label{eq:q_abs_limit_superpos}
\widehat{\left|q^{(N)}(t)\right|} \underset{t \to \pm \infty}{\to} 4  \sum_{r=1}^{N} \eta_r^{\pm1} \exp(\mp 2\sigma_r t) \left|\sum_{k=r}^{N} \sigma_k a_{r,k}\right|	
\end{align}

\end{lemma}

\begin{proof} Consider the signal update rule of the Darboux transformation \eqref{eq:sig_update} in the limits $t \to \pm \infty$. Using Lem. \ref{lem:Darboux_iteration_limit} b), the signal update rule simplifies to
\begin{align}
q^{(N)}(t) & \to -4 \sum_{k=1}^{N} \sigma_k \rho_k^{*(k-1)}(t), \quad t \to \infty \qquad \mathrm{where} \quad \rho_k^{(k-1)}(t) \to 0 \label{eq:q_limit+}	
\\
q^{(N)}(t) & \to -4 \sum_{k=1}^{N} \sigma_k \frac{1}{\rho_k^{(k-1)}(t)}, \quad t \to -\infty,  \qquad \mathrm{where} \quad \rho_k^{(k-1)}(t)  \to \infty.\label{eq:q_limit-}
\end{align}
Using \eqref{eq:a_rk_superpos} in the above equations, we get
\begin{align}\label{eq:q_limit_superpos+}
 q^{(N)}(t) &\underset{t \to \infty}{\to} -4 \sum_{k=1}^{N} \sigma_k \rho_k^{*(k-1)}(t)	
 = -4 \sum_{k=1}^{N} \sigma_k \sum_{r=1}^k a^*_{r,k} \rho_r^{*\, (0)}(t)	
= -4 \sum_{r=1}^{N} \rho_r^{*\, (0)}(t) \sum_{k=r}^N \sigma_k a^*_{r,k} 
\\
 q^{(N)}(t) &\underset{t \to -\infty}{\to} -4 \sum_{k=1}^{N} \sigma_k \frac{1}{\rho_k^{(k-1)}(t)}	
 = -4 \sum_{k=1}^{N} \sigma_k \sum_{r=1}^{k} a_{r,k} \frac{1}{\rho_r^{(0)}(t)}	
= -4 \sum_{r=1}^{N} \frac{1}{\rho_r^{(0)}(t)} \sum_{k=r}^{N} \sigma_k a_{r,k}
\label{eq:q_limit_superpos-}
\end{align}
Note further, that the index $k$ ranges only from $r$ to $N$, since all $a_{r,k}$ with $k<r$ are zero (see \eqref{eq:rho_update_t_infty_general}, \eqref{eq:rho_update_t_-infty_general}). Then, replacing $\rho_r^{(0)}(t)$ with the initialization \eqref{eq:rho_init_etak} leads to
\begin{align}
& q^{(N)}(t) \underset{t \to \infty}{\to}	  -4 \sum_{r=1}^{N} \eta_r \exp(- 2\sigma_r t) \exp(-j\varphi_r-2j\omega_r t) \sum_{k=r}^{N} \sigma_k a_{r,k}^*	\label{eq:q_+limit_superpos2}	
\\
& q^{(N)}(t) \underset{t \to -\infty}{\to} 	 -4 \sum_{r=1}^{N}  \frac{1}{\eta_r} \exp(2\sigma_r t) \exp(-j\varphi_r-2j\omega_r t) \sum_{k=r}^{N} \sigma_k a_{r,k}	\label{eq:q_-limit_superpos2}
\end{align}

It follows from \eqref{eq:q_+limit_superpos2} and \eqref{eq:q_-limit_superpos2} that
\begin{align}
|q^{(N)}(t)| \underset{t \to \pm \infty}{\to} &\left|4 \sum_{r=1}^{N} \eta_r^{\pm 1} \exp(\mp 2\sigma_r t) \exp(\pm j\varphi_r \pm 2j\omega_r t) \sum_{k=r}^{N} \sigma_k a_{r,k}\right|	\label{eq:q_abs_lim}
\\
\leq &4 \sum_{r=1}^{N} \eta_r^{\pm 1} \exp(\mp 2\sigma_r t) \left| \sum_{k=r}^{N} \sigma_k a_{r,k}\right|	\label{eq:q_env_lim}
\end{align}

We consider the envelope of the signal which is the maximum absolute value of the signal for all possible phase combinations $\{\varphi_k\}_{k=1}^N$. For each given time instance $t_0$, there exists one combination of $\{\varphi_k\}$ such that all terms $\exp(-j\varphi_r-2j\omega_r t) \sum_{k=r}^{N} a_{r,k}$ in \eqref{eq:q_abs_lim} have identical phase. Therefore we consider the case $\arg\left\{e^{j\varphi_r+2j\omega_r t_0}\sum_{k=r}^{N} a_{r,k}\right\}
=\arg\left\{e^{j\varphi_s+2j\omega_s t_0}\sum_{k=s}^{N} a_{s,k}\right\}$
for all $r,s=1,...,N$. The signal envelope among all phase combinations $\{\varphi_k\}$ for $t \to \pm \infty$ is \eqref{eq:q_env_lim}.

\end{proof}

\begin{lemma}\label{lem:sig_tails_imag}
Consider an $N-$soliton with imaginary eigenvalues $\lambda_k=j\sigma_k$, $\sigma_k \in \mathbb{R}^+$ and spectral amplitude scaling $\eta_k$. The multi soliton pulse in the tails $t \to \pm \infty$ is given as
\begin{align}
\left|q^{(N)}(t)\right| \to 4  \eta_N^{\pm 1} \exp(\mp 2\sigma_N t) \sigma_N \prod_{k=1}^{N-1}\left|\frac{\sigma_N+\sigma_k}{\sigma_N-\sigma_k}\right| 	
\end{align}

\end{lemma}

\begin{proof} Consider the general signal limits \eqref{eq:q_abs_gen_limit} in Lem. \ref{lem:sig_tails_gen} and limit the eigenvalues to be imaginary $\lambda_k=j\sigma_k$. Under the (practically relevant) assumption of well separated eigenvalues, all terms but the leading exponential term given by the smallest eigenvalue $\sigma_N=\underset{k}{\min}\, \sigma_k$ can be neglected in the limit $t \to \pm \infty$. Consequently, the absolute value of the pulse in \eqref{eq:q_abs_limit_superpos} in its tails simplifies to \eqref{eq:q_abs_limit_superpos_imag} where $a_{N,N}$ can be derived from \eqref{eq:rho_update_t_infty_general}. If $\Delta t_k=\Delta t$ (symmetric pulses), $|q^{(N)}(t)|$ converges to its limit fast; convergence takes longer when $|\Delta t_N-\Delta t_l|$ ($l\neq N$) is large.
\begin{align}\label{eq:q_abs_limit_superpos_imag}
\left|q^{(N)}(t)\right| \underset{t \to \pm \infty}{\to} 4  \eta_N^{\pm 1} \exp(\mp 2\sigma_N t) \sigma_N \left|a_{N,N}\right| \qquad \mathrm{with} \quad a_{N,N} = \prod_{k=1}^{N-1}\frac{\sigma_N+\sigma_k}{\sigma_N-\sigma_k}
\end{align}
\end{proof}

\begin{lemma}\label{lem:sig_tails_real}
Consider a multi-soliton pulse with $N$ eigenvalues parallel to the real axis ${\lambda_k=\omega_k+j\sigma}$ ($\sigma \in \mathbb{R}^+$, $\omega_k \in \mathbb{R}$) and spectral amplitude scaling $\eta_k$. We denote $A_r=\left|\sum_{k=r}^{N}  a_{r,k}\right|$. The soliton pulse in the tails $t \to \pm \infty$ is bounded by \eqref{eq:signal_bound_real} and the signal envelope converges to \eqref{eq:signal_envelope_real}.
\begin{align}
\left|q^{(N)}(t)\right| & \leq 4  \sigma \exp\left(\mp 2\sigma \left(t \mp \frac{1}{2 \sigma}\ln\left(\sum_{r=1}^{N} \eta_r^{\pm 1}  A_r \right) \right) \right) \label{eq:signal_bound_real}	
\\
\widehat{\left|q^{(N)}(t)\right|} & \to 4  \sigma \exp\left(\mp 2\sigma \left(t \mp \frac{1}{2 \sigma}\ln\left(\sum_{r=1}^{N} \eta_r^{\pm 1}  A_r \right) \right) \right) \label{eq:signal_envelope_real}
\end{align}

\end{lemma}

\begin{proof} Consider the general soliton signal \eqref{eq:q_abs_gen_limit} and its envelope \eqref{eq:q_abs_limit_superpos} for $t \to \pm \infty$ in Lem.~ \ref{lem:sig_tails_gen}. By limiting all eigenvalues to have identical imaginary part $\lambda_k=\omega_k+j\sigma$, we can conclude
\begin{align*}
\left|q^{(N)}(t)\right| \leq 4  \sigma \exp(\mp 2\sigma t) \sum_{r=1}^{N} \eta_r^{\pm 1}  \left|\sum_{k=r}^{N}  a_{r,k}\right| 
\qquad \widehat{\left|q^{(N)}(t)\right|} \to 4  \sigma \exp(\mp 2\sigma t) \sum_{r=1}^{N} \eta_r^{\pm 1}  \left|\sum_{k=r}^{N}  a_{r,k}\right|.
\end{align*}

\end{proof}

\begin{proof}[Proof of Thm. \ref{th:T_imag}]
Consider a multi-soliton pulse with imaginary eigenvalues $\lambda_k=j\sigma_k$ in the limits $t \to \pm \infty$ according to Lem. \ref{lem:sig_tails_imag}. Since the signal tails have the same slope for $t \to -\infty$ and $t \to \infty$, the smallest pulse duration $T=T_+ -T_-$ according to Def. \ref{def:TB_def_energy} is attained when the pulse is truncated such that the same fraction $\frac{\varepsilon}{2} E_\mathrm{tot}$ of the total soliton energy $E_\mathrm{tot}=4\sum_{k=1}^N \sigma_k$ is lost in both of the respective tails. Then, the pulse duration can be calculated from  $\frac{\varepsilon}{2} E_\mathrm{tot} =\int_{-\infty}^{T_-} \left|q^{(N)}(t)\right|^2 \partial t = \int_{T_+}^{\infty} \left|q^{(N)}(t)\right|^2 \partial t$. Using the signal tails in Lem. \ref{lem:sig_tails_imag} we get
${\frac{\varepsilon}{2} E_\mathrm{tot} \approxeq \int_{-\infty}^{T_-}16 \frac{\sigma_N^2}{\eta_N^2} \exp\left(4\sigma_N t \right) |a_{N,N}|^2 \partial t
 \approxeq \int_{T_+}^{\infty} 16 \sigma_N^2 \eta_N^2 \exp\left(-4\sigma t \right) |a_{N,N}|^2 \partial t}$.
Solving the equation for $T_+$ and $T_-$ leads to the pulse duration in Th.~\ref{th:T_imag}. Note that the accuracy of this approximation decreases if $\Delta t_k$ are very different and solitonic components become separated (unless for very small $\varepsilon \to 0$), since \eqref{eq:q_abs_limit_superpos_imag} in Lem. \ref{lem:sig_tails_imag} becomes less precise for smaller $|t|$.

\end{proof}

\begin{proof}[Proof of Thm. \ref{th:T_real}]
Consider a multi-soliton with eigenvalues parallel to the real axis $\lambda_k=\omega_k+j\sigma$ in its tails $t \to \pm \infty$ according to Lem. \ref{lem:sig_tails_real}. Since the signal tails have the same slope for $t \to -\infty$ and $t \to \infty$, the smallest pulse duration $T=T_+ -T_-$ according to Def. \ref{def:TB_def_energy} is attained when the pulse is truncated such that the same fraction $\frac{\varepsilon}{2} E_\mathrm{tot}$ of the total soliton energy $E_\mathrm{tot}=4\sum_{k=1}^N \sigma_k=4N\sigma$ is lost in both of the respective tails. Then, the pulse duration can be calculated from \eqref{eq:outage_energy_EVreal} where we use the signal envelope in Lem. \ref{lem:sig_tails_real} to approximate the integral:
\begin{align}\label{eq:outage_energy_EVreal}
\frac{\varepsilon}{2} E_\mathrm{tot} =\int_{-\infty}^{T_-} \left|q^{(N)}(t)\right|^2 \partial t = \int_{T_+}^{\infty} \left|q^{(N)}(t)\right|^2 \partial t
\approxeq 4 \sigma \exp\left(\pm 4\sigma \left(T_\mp \pm \frac{1}{2 \sigma} \ln\left(\sum_{r=1}^N \eta_r^{\mp 1} A_r \right)\right)\right).
\end{align}
Solving \eqref{eq:outage_energy_EVreal} for $T_+$ and $T_-$ leads to the pulse duration in Th.~\ref{th:T_real}. Showing the approximation of $T_\mathrm{max}(\varepsilon)$ in  \eqref{eq:T_real} being minimized for $\eta_r=\eta$, needs the proof of $\left. T_\mathrm{max}(\varepsilon)\right|_{\eta_r} \geq \left. T_\mathrm{max}(\varepsilon)\right|_{\eta_r=\eta}$.
Using \eqref{eq:T_real} 
this rewrites as
${\sum_{r=1}^{N} \eta_r A_r  \cdot \sum_{r=1}^{N} \frac{1}{\eta_r} A_r \geq \left(\sum_{r=1}^{N} A_r \right)^2}$. Since $\eta_r, A_r \in \mathbb{R}$, we have
${\sum_{r=1}^{N} \left|\sqrt{\eta_r A_r}\right|^2  \cdot \sum_{r=1}^{N} \left|\sqrt{\frac{1}{\eta_r} A_r}\right|^2  	
\geq \left|\sum_{r=1}^{N} A_r \right|^2}$
which is the Cauchy-Schwarz inequality.

\end{proof}

\section{Proof of approximations Thm.~\ref{th:B_imag} and Thm.~\ref{th:B_real}}\label{sec:proof_Th34}

\begin{lemma}\label{lem:sep_sol_spectrum}
Consider an $N-$soliton with $\lambda_k=\omega_k+\sigma_k$, temporally separated into its first order soliton components. Its (linear) spectrum becomes \eqref{eq:soliton_spectrum} with the corresponding envelope \eqref{eq:soliton_spectrum_envelope}.
\begin{align}
Q(f) & =-\pi \sum_{k=1}^N \exp\left(j (-\varphi_k-2\omega_k \Delta t_k-2\pi f \Delta t_k) \right) \mathrm{sech}\left(\frac{\pi^2}{2\sigma_k}\left(f+\frac{\omega_k}{\pi}\right)\right) \label{eq:soliton_spectrum}
\\
\widehat{|Q(f)|} & =\pi \sum_{k=1}^N \mathrm{sech}\left(\frac{\pi^2}{2\sigma_k}\left(f+\frac{\omega_k}{\pi}\right)\right)\label{eq:soliton_spectrum_envelope}
\end{align} 

\end{lemma}

\begin{proof} The spectrum of fundamental soliton solution with eigenvalues $\lambda_k=j\sigma_k+\omega_k$, spectral amplitude scaling $\eta_k=\exp(2\sigma_k \Delta t_k)$ and spectral phase $\varphi_k$ is well known to be \eqref{eq:one_soliton_spectrum}.
\begin{equation}\label{eq:one_soliton_spectrum}
Q(f) =-\pi \exp\left(j (-\varphi_k-2\omega_k \Delta t_k-2\pi f \Delta t_k)\right) \mathrm{sech}\left(\frac{\pi^2}{2\sigma_k}\left(f+\frac{\omega_k}{\pi}\right)\right).
\end{equation}
Assuming a multi-soliton being separated into its first order components, it can be represented as a linear superposition of first order solitons \eqref{eq:one_soliton_spectrum} leading to \eqref{eq:soliton_spectrum}.
We can conclude from \eqref{eq:soliton_spectrum} that $|Q(f)| \leq \pi \sum_{k=1}^N \mathrm{sech}\left(\frac{\pi^2}{2\sigma_k}\left(f+\frac{\omega_k}{\pi}\right)\right)$.
Furthermore, for each frequency point $f_0$, there exists one combination of $\{\varphi_k\}$ such that all terms ${\exp(-j\varphi_k-2j\omega_k \Delta t_k-2j\pi f \Delta t_k)}$ in \eqref{eq:soliton_spectrum} have identical phase. Thus, the envelope of the spectrum among all phase combinations $\{\varphi_k\}$ is \eqref{eq:soliton_spectrum_envelope}.

\end{proof}

\begin{proof}[Proof of Thm. \ref{th:B_imag}]
Consider the signal spectrum \eqref{eq:soliton_spectrum} in Lem.~\ref{lem:sep_sol_spectrum} for imaginary eigenvalues $\lambda_k=j\sigma_k$. In the limits $f\to \pm \infty$, $\mathrm{sech}(x)$ decays exponentially. Thus all but the leading exponential term, determined by $\sigma_1=\underset{k}{\max} \, \sigma_k$, can be neglected in \eqref{eq:soliton_spectrum_envelope}. Thus we get
\begin{equation}\label{eq:Q_imag_limit}
|Q(f)| \to 2\pi \exp \left(\mp \frac{\pi^2 f}{2 \sigma_1} \right), \quad f \to \pm \infty
\end{equation}
Using \eqref{eq:Q_imag_limit} and noting its symmetry, the bandwidth $B$ according to Def.~\ref{def:TB_def_energy} can be estimated from\\
${\frac{\varepsilon}{2}E_\mathrm{tot} = \int_{B/2}^\infty |Q(f)|^2 \partial f 
\approxeq 
4 \sigma_1 \exp(-\frac{\pi^2 B}{2 \sigma_1})}$.
Solving for $B$ leads to the approximation in Th.~\ref{th:B_imag}.

\end{proof}

\begin{proof}[Proof of Thm. \ref{th:B_real}]
Consider the envelope spectrum \eqref{eq:soliton_spectrum_envelope} in Lem.~\ref{lem:sep_sol_spectrum} for eigenvalues parallel to the real axis $\lambda_k=\omega_k + j\sigma$. In the limits $f \to \pm \infty$, the envelope spectrum writes as
\begin{align}
\widehat{|Q(f)|} & \to 2\pi \exp\left(\mp \frac{\pi^2 f}{2 \sigma}\right)\sum_{k=1}^N \exp\left(\mp \frac{\pi \omega_k}{2\sigma}\right), \quad f \to \pm \infty.
\end{align} 
We can derive an approximation of the bandwidth $B=B_+ - B_-$ according to Def.~\ref{def:TB_def_energy} from $\frac{\varepsilon}{2	}E_\mathrm{tot} \approxeq \int_{-\infty}^{B_-} \widehat{|Q(f)|}^2 \partial f =\int_{B_+}^{\infty} \widehat{|Q(f)|}^2 \partial f 
\approxeq 4 \sigma \exp(\mp \frac{\pi^2 B_\pm}{\sigma}) \left(\sum_{k=1}^N \exp\left(\mp\frac{\pi \omega_k}{2\sigma}\right)\right)^2$.
Solving for $B_+$ and $B_-$ leads to the bandwidth approximation in Thm.~\ref{th:B_real}.

\end{proof}

\ifCLASSOPTIONcaptionsoff
  \newpage
\fi

\bibliographystyle{IEEEtranTCOM}
\bibliography{references_nft}

\end{document}